\documentclass[journal]{IEEEtran}
\usepackage{graphicx}
\usepackage{booktabs}
\usepackage{cite}
\usepackage{capt-of}
\usepackage{algorithm}
\usepackage{algpseudocode}

\usepackage{amsmath,amssymb,amsfonts,bm}
\usepackage{subcaption}
\usepackage{booktabs}
\usepackage{textcomp}
\usepackage{xcolor}
\usepackage{multirow}

\usepackage{tikz}
\usetikzlibrary{decorations.pathreplacing}   
\newtheorem{proposition}{Proposition}
\newtheorem{theorem}{Theorem}
\newtheorem{assumption}{Assumption}
\newtheorem{remark}{Remark}
\newtheorem{lemma}{Lemma}

\long\def\comment#1{}

\newfont{\bbb}{msbm10 scaled 700}

\newfont{\bb}{msbm10 scaled 1100}

\renewcommand{\det}{{\hbox{det}}}

\renewcommand{\arg}{{\hbox{arg}}}

\begin{document}

\title{Positioning a Movable Antenna Without CSI: A Kernelized Bandit Under Costly Movement}

\author{Wonseok Choi,~\IEEEmembership{Graduate Student Member,~IEEE,} Yunseob Tae,~\IEEEmembership{Graduate Student Member,~IEEE,} \\Jeongjae Lee,~\IEEEmembership{Member,~IEEE,} and~Songnam~Hong,~\IEEEmembership{Senior Member,~IEEE}
        \thanks{W. Choi and Y. Tae contributed equally to this work.}
        \thanks{W. Choi, Y. Tae, and S. Hong are with the Department of Electronic Engineering, Hanyang University, Seoul, Korea. (e-mail: \{ryan4975, tys7524, snhong\}@hanyang.ac.kr).}

        \thanks{J. Lee is with the Department of Electronic Engineering, Hanyang University, Seoul, Korea, and the Ming Hsieh Department of Electrical Engineering, University of Southern California, Los Angeles, CA, USA. (e-mail: jl\_939@usc.edu).}
        
}

\maketitle

\begin{abstract}
Movable antenna (MA) systems reconfigure the wireless channel by repositioning the antenna within a confined region. The position optimization behind this capability has been studied under the premise that the channel at every candidate position is known. However, acquiring that knowledge is costly: a position can be measured only after the antenna has moved there, which takes many slots at the speed of the actuator, and the channel drifts meanwhile. We formulate MA positioning as a non-stationary kernelized bandit with a reachability constraint, in which exploration, tracking, and the reward lost in transit are coupled through a single physical action. We propose MoveUCB, which addresses the reachability constraint through a global target search, a persistence rule that holds the target across slots, and a transit-aware movement cost term. We prove that MoveUCB attains sublinear dynamic regret under the standard variation-budget condition of non-stationary bandits alone, with the physical parameters entering only through constants. Simulations against a dynamic-programming oracle and movement-unaware baselines show that MoveUCB attains the lowest regret with about half the travel of the closest baseline that moves.
\end{abstract}


\section{Introduction}
\label{sec:intro}

In conventional antenna systems the propagation channel is given rather than chosen: the antenna occupies a fixed location, and the small-scale fading realized at that location cannot be influenced by the transceiver. Movable antenna (MA) systems~\cite{zhu2024movable} relax this constraint by mechanically repositioning the antenna within a confined region. Since the fading pattern decorrelates over distances on the order of a wavelength, even centimeter-scale repositioning can convert a deep fade into a favorable channel realization, and this additional spatial degree of freedom has been exploited for beamforming and rate maximization~\cite{ma2024multibeam,feng2024weighted}. Closely related fluid antenna systems~\cite{wong2021fluid,cheng2024sumrate} pursue the same spatial reconfiguration by electronically switching among predefined ports, and much of the position-optimization machinery is shared between the two.

The mechanical nature of MA actuation, however, sets it apart from a port switch. Whereas a switch is instantaneous, a repositioning proceeds at the speed of the actuator: the antenna traverses the region over many slots, transmitting from intermediate positions rather than from its destination, and drawing energy from the motor-driven mechanism along the way~\cite{ning2025movable}. As carrier frequencies rise and coherence times shrink, this overhead consumes a non-negligible fraction of the transmission time. A recent line of work has therefore incorporated the movement cost into the optimization itself: minimizing the repositioning delay for given target positions~\cite{li2026trajectory}, maximizing throughput or energy efficiency under movement-delay penalties~\cite{wang2026throughput,ding2026energy}, and, in our own prior work, movement-delay-aware sum-rate maximization for a linear MA array at the base station~\cite{tae2026delay}. These formulations, and much of the position optimization literature more broadly, share a common premise: the channel is known before the optimization is carried out, and hence so is the reward of every candidate antenna position.

This premise is difficult to sustain in practice: the reward at a new position can be measured only by relocating the antenna there. Obtaining channel state information (CSI) across the spatial domain via pilots incurs prohibitive overhead, which has motivated CSI-free position optimization that works directly from received-signal measurements~\cite{zeng2025csi}. This method treats the position as the variable of a static optimization and refines it by zeroth-order gradient ascent, with the number of measurements as the figure of merit. The present setting differs from this static optimization in two respects. First, the reward drifts while the antenna is in transit, so a measurement becomes outdated before the next one is collected. Second, each measurement requires a relocation that is limited by the actuator speed and that costs transmission time. Position selection is therefore not a one-shot optimization but a sequential decision problem: the transceiver observes a noisy reward only at its current position, and must decide, based on this history, where to explore next despite the cost and delay of getting there.

The learning literature offers no ready solution to this problem. Kernelized bandits~\cite{srinivas2010,chowdhury2017} exploit spatial correlation to optimize an unknown function from point evaluations, but assume a stationary environment and free sampling of arbitrary points. Non-stationary bandit formulations~\cite{besbes2019,cheung2019} handle temporal drift through variation budgets or sliding windows, and kernelized versions~\cite{bogunovic2016} forget old observations by resetting or by discounting, yet all still permit the learner to play any arm at any slot. Online learning with switching costs~\cite{CesaBianchi2013} and its metric extensions~\cite{koren2017nips} penalize decision changes in proportion to the distance moved, but the penalty is an abstract cost rather than a physical constraint. None of these formulations captures the reachability limit that confines each decision to a neighborhood of the current one. Even a highly valuable position cannot be reached within a single slot if it is far away; the antenna must instead approach it gradually, over several slots. The MA position optimization problem lies at the intersection of all three---bandit feedback, non-stationary temporal drift, and physically constrained, costly actuation---an intersection that, to the best of our knowledge, remains unaddressed.

To bridge this gap, our main contributions are summarized as follows:
\begin{itemize}

\item We formulate online MA position learning as a non-stationary kernelized bandit under a reachability constraint, so that exploration, tracking, and the reward lost in transit are coupled through a single action. Performance is measured by the dynamic regret against an oracle subject to the same constraint and cost.

\item We propose {\bf MoveUCB}, which combines a physics-informed kernel with a movement cost term, a \emph{persistence rule} that keeps the target fixed until reached, and a \emph{two-stage selection}: a standard UCB rule confines the target to what is reachable in one slot and so never evaluates a distant position, whereas MoveUCB selects the target over the whole region and approaches it under this rule.

\item We prove that MoveUCB attains sublinear regret under the standard variation-budget condition alone. Since optimism controls the reward but not the distance, we introduce a \emph{reverse chain decomposition} of the information gain to bound the travel. The decomposition applies to any sliding-window kernelized bandit and may be of independent interest.

\item We compare MoveUCB against a dynamic-programming oracle, a movement-unaware kernelized UCB, and sweep-based repositioning, each tuned to its own optimum, and show that it attains the lowest regret at about half the travel of the closest competitor. An ablation confirms that the two-stage selection drives this gain, while the movement cost term and the persistence rule keep the search affordable at no measurable cost in regret.

\end{itemize}

The remainder of this paper is organized as follows. Section~\ref{sec:model} describes the system model and formulates the problem. Section~\ref{sec:algorithm} presents the proposed MoveUCB algorithm, and Section~\ref{sec:analysis} establishes its regret guarantee. Section~\ref{sec:sim} reports the simulation results, and Section~\ref{sec:conclusion} concludes the paper.

\section{System Model and Problem Formulation}
\label{sec:model}

This section describes the physical model of a movable antenna (MA) system and casts MA position optimization as an online learning problem with a reachability constraint.

\subsection{System Model}
\label{subsec:sysmodel}

We consider a point-to-point communication scenario in which the receiver is equipped with a single MA that can be positioned along a linear region $\mathcal{D}\triangleq[0,D]$ of length $D$. Time is divided into slots of duration $T_{\mathrm{s}}$, indexed by $t\in\{1,2,\dots\}$.
The antenna position during slot $t$ is denoted by $x_{t}\in\mathcal{D}$, with a given initial position $x_{0}\in\mathcal{D}$ before the first slot. As illustrated in Fig.~\ref{fig:slot}, each slot comprises an actuation interval of duration $T_{\mathrm{a}}$, during which the antenna may be repositioned by its mechanical actuator of maximum speed $v$, and a transmission interval of duration $T_{\mathrm{s}}-T_{\mathrm{a}}$, during which the data transmission operates from the resulting position. With the \emph{stroke} $\rho_{\max}\triangleq vT_{\mathrm{a}}$ denoting the maximum displacement achievable in one slot, the position update is constrained to the reachable region
\begin{equation}
  x_{t}\in\mathcal{A}(x_{t-1})\triangleq
  \big[x_{t-1}-\rho_{\max},\,x_{t-1}+\rho_{\max}\big]\cap\mathcal{D}.
  \label{eq:reachable}
\end{equation}
The actuation interval $T_{\rm a}$ is provisioned for the maximum displacement $\rho_{\max}$ and held fixed across slots. {A policy may select a destination $z\in\mathcal D$ outside the
reachable region, $z\notin\mathcal A(x_{t-1})$. Reaching it then takes
$\lceil|z-x_{t-1}|/\rho_{\max}\rceil$ slots, on each of which the
antenna is \emph{in transit}: it advances toward $z$ as far as
\eqref{eq:reachable} permits and transmits from that intermediate
position rather than from $z$.}

\begin{figure}[t]
\centering
\includegraphics[width=0.85\columnwidth]{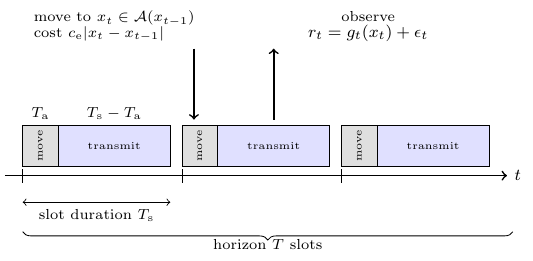}
\caption{Slot structure: repositioning within $T_{\mathrm{a}}$, subject to~\eqref{eq:reachable}, followed by transmission and one noisy reward observation.}
\label{fig:slot}
\end{figure}

Let $g_{t}(x)$ denote the reward obtained when the antenna is placed at position $x$ during slot $t$, e.g., the achievable rate
\begin{equation}
  g_{t}(x)=\bigg(1-\frac{T_{\mathrm{a}}}{T_{\mathrm{s}}}\bigg)
  \log_{2}\!\bigg(1+\frac{P\,|h_{t}(x)|^{2}}{\sigma_{n}^{2}}\bigg),
  \label{eq:reward}
\end{equation}
where $h_{t}(x)$ is the channel at position $x$ and slot $t$, $P$ is the transmit power, and $\sigma_{n}^{2}$ is the noise power. The common factor $1-T_{\mathrm{a}}/T_{\mathrm{s}}$ is the ratio of the transmission interval to the slot duration. The slot duration is assumed shorter than the channel coherence time, so that $h_{t}(\cdot)$ is a single realization within slot $t$, whereas the channel geometry evolves across slots. In each slot, the receiver observes only a noisy reward at its current antenna position,
\begin{equation}
  r_{t}=g_{t}(x_{t})+\epsilon_{t},
  \label{eq:observation}
\end{equation}
where $\epsilon_t$ accounts for the finite-sample estimation error of the reward over the transmission interval. Following the standard observation model in kernelized bandits~\cite{chowdhury2017}, $\epsilon_t$ is assumed to be zero-mean and conditionally sub-Gaussian given the history, as formalized in Assumption~\ref{as:noise}. Once the transmission interval of slot $t$ concludes, the receiver holds the history $\mathcal{F}_{t}=\{(x_{\tau},r_{\tau})\}_{\tau\le t}$, i.e., the positions it has visited and the noisy rewards collected, and must select the next position $x_{t+1}\in\mathcal{A}(x_{t})$ based on $\mathcal{F}_{t}$ alone.

\subsection{Problem Formulation}
\label{subsec:problem}

The receiver seeks a positioning policy $\pi$ that maximizes the reward accumulated over the horizon net of the actuation energy,
\begin{align}
  \max_{\pi}\ &\mathbb{E}\!\left[\sum_{t=1}^{T}\Big(g_{t}(x_{t})-c_{\mathrm{e}}\big|x_{t}-x_{t-1}\big|\Big)\right]
  \nonumber\\
  \text{s.t.}\ &x_{t}=\pi(\mathcal{F}_{t-1})\in\mathcal{A}(x_{t-1}),
  \label{eq:online}
\end{align} where repositioning the antenna incurs an energy penalty of $c_{\mathrm{e}}|x_{t}-x_{t-1}|$, the coefficient $c_{\mathrm{e}}>0$ converting the actuation energy per unit displacement into the units of the rate reward~\cite{ding2026energy}. Problem~\eqref{eq:online} cannot be solved directly, since $\{g_{t}\}$ is revealed only causally and then only through the noisy observations in~\eqref{eq:observation}.

As an upper bound on the achievable performance, we consider an \emph{oracle} that knows the entire reward sequence $\{g_{t}\}_{t=1}^{T}$ in advance while being subject to the same
reachability constraint. Its optimal cumulative reward is
\begin{equation}
  \mathrm{OPT}_{T} \triangleq
  \max_{\substack{\{x_{t}\}_{t=1}^{T}:\\ x_{t}\in\mathcal{A}(x_{t-1})}}
  \ \sum_{t=1}^{T}\Big(g_{t}(x_{t})-c_{\mathrm{e}}\big|x_{t}-x_{t-1}\big|\Big).
  \label{eq:oracle}
\end{equation}
Accordingly, we evaluate $\pi$ through its {\it dynamic regret} relative to the oracle,
\begin{equation}
  \mathcal{R}_{T}(\pi)
  \triangleq
  \mathrm{OPT}_{T}
  -\mathbb{E}\!\left[\sum_{t=1}^{T}
  \Big(g_{t}(x_{t})-c_{\mathrm{e}}\big|x_{t}-x_{t-1}\big|\Big)\right],
  \label{eq:regret}
\end{equation}
where the expectation is over the observation noise and any randomization in the policy. Our objective is to develop an online positioning policy that minimizes $\mathcal{R}_{T}(\pi)$.


\section{Proposed MoveUCB}
\label{sec:algorithm}

Three mechanisms drive the regret in \eqref{eq:regret}: (i) insufficient exploration, in which a favorable position goes unidentified because its reward is not sufficiently sampled; (ii) insufficient tracking, in which the antenna remains at a position whose reward has since degraded; and (iii) excessive movement, in which the reward gain at the destination does not compensate for the reward lost in transit and the actuation energy consumed. Mechanisms (i) and (ii) arise in stationary and non-stationary bandit problems alike, whereas mechanism (iii) is specific to a decision variable that must physically move, and~\eqref{eq:regret} penalizes it whenever the learner travels more than the oracle does. The proposed algorithm, referred to as Movement-Aware Upper Confidence Bound (\textbf{MoveUCB}), addresses these mechanisms through four components, each derived from the regret analysis of Section~\ref{sec:analysis}: a physics-informed kernel, a movement cost term, a persistence rule, and a two-stage selection. We describe each in turn.

\subsection{Reward Estimation with a Physics-Informed Kernel}
\label{subsec:kernel}

The history $\mathcal{F}_{t}$ must be exploited along two axes: spatial correlation allows an observation at $x_{\tau}$ to inform nearby positions, so that the reward over the continuum can be interpolated from finitely many samples, while temporal variation makes old observations outdated. The reward is therefore estimated by kernel ridge regression over a sliding window of the $w$ most recent observations, adapting the standard estimator in kernelized
bandits~\cite{srinivas2010,chowdhury2017}. Let $\mathcal{W}_{t}=\{\tau:\max(1,t-w+1)\le\tau\le t\}$ denote the active window, including the observation just collected at $x_{t}$. A Gaussian process prior $\mathcal{GP}(0,k)$ is placed on $g_{t}$ and conditioned on the windowed observations, which yields a point estimate and its uncertainty at every position, the two quantities required by the selection rule of Section~\ref{subsec:selection}. With the kernel vector $\mathbf{k}_{t}(x)=[k(x,x_{\tau})]_{\tau\in\mathcal{W}_{t}}$, the Gram matrix $\mathbf{K}_{t}=[k(x_{\tau},x_{\tau'})]_{\tau,\tau'\in\mathcal{W}_{t}}$, and the observation vector $\mathbf{r}_{t}=[r_{\tau}]_{\tau\in\mathcal{W}_{t}}$, the posterior mean and variance are given by
\begin{align}
  \mu_{t}(x)
  &=\mathbf{k}_{t}(x)^{\mathrm{T}}
    \big(\mathbf{K}_{t}+\varrho\mathbf{I}\big)^{-1}\mathbf{r}_{t},
  \label{eq:posterior_mean}\\
  \sigma_{t}^{2}(x)
  &=k(x,x)-\mathbf{k}_{t}(x)^{\mathrm{T}}
    \big(\mathbf{K}_{t}+\varrho\mathbf{I}\big)^{-1}\mathbf{k}_{t}(x),
  \label{eq:posterior_var}
\end{align}
where $\varrho>0$ serves as a ridge regularization parameter~\cite{rasmussen2006}. Both $\mu_{t}$ and $\sigma_{t}$ are defined over the entire region $\mathcal{D}$, with $\sigma_{t}(x)$ small near recently visited positions and large in unexplored areas. The window length $w$ balances the estimation noise against the temporal staleness of the observations, a trade-off analyzed in Section~\ref{sec:analysis}.

The estimate~\eqref{eq:posterior_mean} is a linear combination of the functions $k(\cdot,x_{\tau})$, so the kernel fixes the basis in which $g_{t}$ is represented. If the basis does not oscillate at the same rates as $g_{t}$, the reproducing kernel Hilbert space (RKHS) norm $\|g_{t}\|_{\mathcal{H}_{k}}$ becomes large, which in turn inflates the confidence parameter of the regret bound. We therefore identify those rates first. Consider a channel composed of $L$ propagation paths under the field-response model~\cite{zhu2023modeling}, each represented by its mean angle of arrival:
\begin{equation}
  h(x)=\sum_{\ell=1}^{L}\alpha_{\ell}\,e^{\jmath\frac{2\pi}{\lambda}x\cos\bar{\theta}_{\ell}},
  \label{eq:multipath}
\end{equation}
where $\alpha_{\ell}$ and $\bar{\theta}_{\ell}$ denote the complex gain and the mean angle of arrival of the $\ell$-th path, and $\lambda$ is the carrier wavelength. The reward~\eqref{eq:reward} depends on position only through $|h(x)|^{2}$, which expands as
\begin{equation}
|h(x)|^{2} =\sum_{\ell=1}^{L}|\alpha_{\ell}|^{2} +\;2\!\!\sum_{\ell<\ell'}\!\!|\alpha_{\ell}\alpha_{\ell'}|
    \cos\!\Big(\tfrac{2\pi}{\lambda}\Delta_{\ell\ell'}\,x +\phi_{\ell\ell'}\Big),
  \label{eq:power_profile}
\end{equation}
where $\Delta_{\ell\ell'}\triangleq\cos\bar{\theta}_{\ell}-\cos\bar{\theta}_{\ell'}$ and $\phi_{\ell\ell'}\triangleq\arg(\alpha_{\ell}\alpha_{\ell'}^{*})$.

The oscillating terms are indexed by \emph{differences} of directional cosines rather than by the cosines themselves, and their frequencies therefore lie in $[0,2/\lambda]$. The finest fringe has a period of $\lambda/2$, which arises when two paths arrive from opposite directions. When a single path dominates, only the constant term survives in \eqref{eq:power_profile}. The logarithm in~\eqref{eq:reward} adds harmonics at sums and differences of these frequencies, but such harmonics decay quickly. The band that the kernel must cover is therefore $[0,2/\lambda]$. Moreover, a physical path is a cluster spanning a small angular spread rather than a plane wave. Each fringe therefore averages nearby frequencies and decorrelates over a coherence length inversely proportional to that spread. The aggregate power, being independent of the relative phases, varies over a much longer scale.

Based on these observations, we design the physics-informed kernel to provide a nonoscillatory component for the aggregate power, a second one for the low-frequency residual introduced by the logarithm, and one damped sinusoid per fringe, which yields
\begin{align}
  k(x,x')&=\sigma_{c}^{2} +\sigma_{0}^{2}\,e^{-\frac{(x-x')^{2}}{2\varsigma_{0}^{2}}}\nonumber\\
  &+\sum_{q=1}^{Q}\sigma_{q}^{2}\,
    \cos\!\bigg(\frac{2\pi(x-x')}{\lambda_{q}}\bigg)\,
    e^{-\frac{(x-x')^{2}}{2\varsigma_{q}^{2}}},
  \label{eq:kernel}
\end{align}
where $\sigma_{c}^{2}$ weights the constant term, which reproduces the aggregate power exactly as the kernel of the one-dimensional space of constants, and $\lambda_{q}$ and $\varsigma_{q}$ are the period and the damping length of the $q$-th fringe. The Gaussian damping is adopted for its smoothness and its exponentially decaying spectrum, both of which the analysis relies on. Each summand is positive definite as a product of a cosine and a squared exponential kernel, and so is~\eqref{eq:kernel}. We denote its total output scale by
\begin{equation}
 \sigma^{2}_{f}\triangleq\sup_{x\in\mathcal{D}}k(x,x) =\sigma_{c}^{2}+\sigma_{0}^{2}+\sum_{q=1}^{Q}\sigma_{q}^{2},
  \label{eq:outputscale}
\end{equation}
which sets the amplitude with which a member of $\mathcal{H}_{k}$ varies across the region. Without the nonoscillatory terms the kernel places no mass at zero frequency, where the reward carries most of its energy, and without the oscillating terms it cannot represent the fringes. The squared exponential, the default choice in kernelized bandits~\cite{srinivas2010,chowdhury2017}, lacks the latter and is recovered from~\eqref{eq:kernel} by setting $Q=0$ and $\sigma_c^2=0$.


The kernel is thus specified by two quantities that are stable over many coherence times---the angular sector $[\theta_{\min},\theta_{\max}]$ and the angular spread $\sigma_{\mathrm{as}}$ of a cluster---together with the output scale $\sigma_{f}^{2}$ measured during the warm-up. The sector bounds the directional-cosine difference of any two paths by $\delta_{\max}=\cos\theta_{\min}-\cos\theta_{\max}$, and the fringe periods are therefore placed at $\lambda_{q}=\lambda/\delta_{q}$ with the $\delta_{q}$ spanning $[0,\delta_{\max}]$. The spread blurs each fringe frequency by a proportional amount, which limits how far the fringe stays coherent and gives the common damping length $\varsigma_{q}=\Theta(\lambda/\sigma_{\mathrm{as}})$, whereas $\varsigma_{0}$ needs only to span the residual low-frequency content and is a few multiples of the wavelength. The weights are fixed fractions of $\sigma_{f}^{2}$, which is the sample variance of the warm-up rewards net of $\sigma_{\epsilon}^{2}$. None of the kernel hyperparameters requires knowledge of the instantaneous channel.

The proposed kernel satisfies two mathematical properties required by the regret analysis, both by construction. It is twice differentiable at the origin, and its spectral density is a point mass at the origin plus $2Q+1$ Gaussian lobes with centers confined to $[-2/\lambda,2/\lambda]$. Section~\ref{sec:analysis} uses the differentiability to control the grid discretization error and the confined spectrum to bound the maximum information gain.

\subsection{Movement-Aware Position Optimization}
\label{subsec:selection}

The estimate and its uncertainty in~\eqref{eq:posterior_mean}--\eqref{eq:posterior_var} are combined into the upper confidence bound (UCB)
\begin{equation}
  \mathrm{UCB}_{t}(x)\triangleq\mu_{t}(x)+\beta\,\sigma_{t}(x),
  \label{eq:ucb}
\end{equation}
where $\beta>0$ is a confidence parameter. 
With an appropriate $\beta$, $\mathrm{UCB}_{t}(x)$ upper-bounds $g_{t}(x)$ with high probability, providing an optimistic estimate of the reward attainable at $x$. A larger $\beta$ weights the uncertainty more heavily and thus promotes exploration. While~\eqref{eq:ucb} is the standard acquisition function in kernelized bandits~\cite{srinivas2010,chowdhury2017}, it does not account for the physical constraints of a movable antenna: a relocation incurs reward loss in transit, and a distant target is reached only after several slots. Taking these constraints into account, the following subsections incorporate the reward loss into the acquisition function, suppress unnecessary target changes in transit, and separate the choice of target from the step taken toward it.

\subsubsection{Accounting for a relocation cost}
\label{subsubsec:price}

A displacement of $|x-x_{t}|$ incurs three costs, all proportional to the distance. The first is the actuation energy $c_{\mathrm{e}}|x-x_{t}|$ of Section~\ref{subsec:problem}. The second is the reward lost while in transit. The actuator covers at most $\rho_{\max}$ per slot, so a journey of length $|x-x_{t}|$ takes $|x-x_{t}|/\rho_{\max}$ slots, during which the antenna transmits from intermediate positions rather than from the destination. Two positions in the region differ in reward by an amount on the order of $\sigma_{f}$, the output scale~\eqref{eq:outputscale} of the kernel, and each transit slot therefore costs $\Theta(\sigma_{f})$ relative to the destination. The third is the staleness of the estimate, as $g_{t}$ drifts over those same $|x-x_{t}|/\rho_{\max}$ slots and a distant target is evaluated under a reward function that has since changed.

Only the first of the three is an energy expenditure; the other two arise because the actuator is slow. The two kinds also enter the objective differently: the energy appears in the explicit term of~\eqref{eq:regret}, whereas the transit-induced costs are reflected through $g_{t}(x_{t})$, since the antenna collects the reward of the position it occupies rather than that of its target. A per-slot acquisition function, which compares candidates under $g_{t}$ alone, cannot capture these two costs, and they must be incorporated explicitly as a function of the distance. Both scale as $\sigma_{f}/\rho_{\max}$ per unit distance, since they accrue over the $1/\rho_{\max}$ slots that each unit of travel occupies. We therefore combine them under a single weight $\eta$ and define the movement cost coefficient
\begin{equation}
  \tilde{c}\;\triangleq\;\frac{1}{w}\Big(c_{\mathrm{e}}
  +\frac{\eta\,\sigma_{f}}{\rho_{\max}}\Big),
  \label{eq:price}
\end{equation}
whose first term is the actuation energy and whose second term collects the two transit-induced costs under a constant weight $\eta>0$. The acquisition function of MoveUCB is then
\begin{equation}
  F_{t}(x)\;\triangleq\;\mathrm{UCB}_{t}(x)-\tilde{c}\,|x-x_{t}|.
  \label{eq:acq-fn}
\end{equation}
The window length $w$ serves as the divisor for both terms, since an observation is treated as informative for $w$ slots: it is over this interval that a relocation returns its benefit, and over the same interval that the estimate motivating the relocation decays. All quantities in~\eqref{eq:price} are already available, as $\sigma_{f}$ is the kernel output scale~\eqref{eq:outputscale} and $\rho_{\max}$ and $w$ are inputs to Algorithm~\ref{alg:moveucb}. In particular, no online estimate of the drift rate is required.

\subsubsection{Target persistence}
\label{subsubsec:hysteresis}

Optimism alone does not bound the movement, because the estimator forgets. The window keeps only the $w$ most recent observations, so an observation collected more than $w$ slots ago is no longer included in the estimate, and the uncertainty $\sigma_t$ at that position grows back. Its upper confidence bound may rise accordingly, not because the position has become better but because uncertainty about its reward has increased. Any position the antenna leaves can therefore become attractive again, and an acquisition function maximized afresh at every slot can send the antenna back to it. This becomes a problem in its own right when a journey takes many slots. The antenna needs approximately $|x-x_{t}|/\rho_{\max}$ slots to reach a target, during which the uncertainty elsewhere can regrow. If a second position becomes the maximizer before the first is reached, the antenna may reverse direction before reaching its target, and repeated redirection can leave it traveling continuously without ever occupying a position it selected.

We therefore make the target persistent. Holding it unconditionally would prevent the antenna from responding to a position that has genuinely become better, and the target $z_{t-1}$ of the previous slot is instead re-evaluated under the current acquisition function $F_{t}$ and replaced only if some position is better by a margin $\Delta>0$,
\begin{equation}
  z_{t}=
  \begin{cases}
    z_{t-1}, & \text{if }\;
      F_{t}(z_{t-1})\;\ge\;\displaystyle\max_{x\in\mathcal{D}}F_{t}(x)
      -\Delta,\\[6pt]
    \displaystyle\arg\max_{x\in\mathcal{D}}F_{t}(x), & \text{otherwise,}
  \end{cases}\label{eq:hysteresis}
\end{equation} with $z_{0}\triangleq x_{0}$. A lead of a small margin, the typical effect of one slot of drift or variance regrowth, is what would redirect the antenna mid-journey, and~\eqref{eq:hysteresis} suppresses such redirection while replacing a target that has genuinely fallen behind. The margin also bounds the loss incurred by holding: since~\eqref{eq:hysteresis} retains a target only while it is within $\Delta$ of the best available position, holding it loses at most $\Delta$ per slot in acquisition value. The two effects appear as the $\Delta T$ term and the travel term of Theorem~\ref{thm:regret}, and Section~\ref{subsec:main} balances them.


\begin{remark}[On bounding the travel]
\label{rem:budget}
Two natural alternatives to~\eqref{eq:hysteresis} suggest themselves. A \emph{travel budget} freezes the antenna once a fixed per-epoch allowance is exhausted, which bounds the movement in every realization but forfeits up to $g_{\max}$ on every frozen slot as the channel drifts away; it saves the energy of a stroke at the expense of the reward of a slot, and is beneficial only when $c_{\mathrm{e}}\rho_{\max}\gtrsim g_{\max}$. A \emph{distance penalty}, as in switching-cost and movement-cost bandits~\cite{CesaBianchi2013,koren2017nips}, avoids the freeze but does not bound the travel on its own, since it must vanish with the horizon whereas the exploration bonus that forgetting restores does not. The proposed persistence rule instead limits how often the target may change, which bounds the travel through the switch count of Lemma~\ref{lem:switch} without ever immobilizing the antenna.
\end{remark}


\begin{algorithm}[t]
\caption{Movement-Aware UCB (MoveUCB)}
\label{alg:moveucb}
\begin{algorithmic}[1]
\Require $\rho_{\max}$, $c_{\mathrm{e}}$, $k$, $w$, $\eta$, $\Delta$,
  $\beta$, $\varrho$, $\mathcal{X}_{g}$, $x_{0}\in\mathcal{X}_{g}$
\State $\mathcal{W}\gets\emptyset$,\quad $z\gets x_{0}$,\quad
  $\tilde{c}\gets\big(c_{\mathrm{e}}+\eta\sigma_{f}/\rho_{\max}\big)/w$
\For{$t=1,2,\dots,T$}
  \State Observe $r_{t}=g_{t}(x_{t})+\epsilon_{t}$; update
    $\mathcal{W}$ to the $w$ most recent samples
  \State Compute $\mu_{t}$, $\sigma_{t}$ on $\mathcal{X}_{g}$
    via~\eqref{eq:posterior_mean}--\eqref{eq:posterior_var}
  \State $F_{t}(x)\gets\mu_{t}(x)+\beta\sigma_{t}(x)
    -\tilde{c}\,|x-x_{t}|$ on $\mathcal{X}_{g}$
    \Comment{\eqref{eq:acq-fn}}
  \State $z^{\star}\gets\arg\max_{x\in\mathcal{X}_{g}}F_{t}(x)$
  \If{$F_{t}(z)<F_{t}(z^{\star})-\Delta$}
    $z\gets z^{\star}$
  \EndIf
  \Comment{\eqref{eq:hysteresis}}
  \State $x_{t+1}\gets x_{t}+\operatorname{sgn}(z-x_{t})
    \min\big(\rho_{\max},|z-x_{t}|\big)$
    \Comment{\eqref{eq:approach}}
\EndFor
\end{algorithmic}
\end{algorithm}

\subsubsection{Two-stage selection}
\label{subsubsec:twostage}

With the target $z_{t}$ of~\eqref{eq:hysteresis} in hand, the antenna advances toward it as far as the actuator permits:
\begin{equation}
  x_{t+1}=x_{t}+\operatorname{sgn}(z_{t}-x_{t})\cdot
    \min\big(\rho_{\max},\,|z_{t}-x_{t}|\big),
  \label{eq:approach}
\end{equation}
which is the projection of $z_{t}$ onto the reachable interval $\mathcal{A}(x_{t})$ of~\eqref{eq:reachable}. The separation of the two stages is what enables the selection of a distant target. Maximizing the acquisition function directly over $\mathcal{A}(x_{t})$ would evaluate only positions within one slot of travel, and the antenna could then neither set out toward a distant peak nor escape a local one. In~\eqref{eq:hysteresis} the maximization ranges over the entire region $\mathcal{D}$, and only the step~\eqref{eq:approach} is subject to the actuator limit. The separation also underlies the analysis: the optimism argument behind Theorem~\ref{thm:regret} compares the selected target against an arbitrary fixed position in $\mathcal{D}$, which is admissible only because~\eqref{eq:hysteresis} maximizes over the entire region rather than over $\mathcal{A}(x_{t})$. The maximization in~\eqref{eq:hysteresis} is one-dimensional and is solved by exhaustive evaluation on the uniform grid $\mathcal{X}_{g}\triangleq\{0,\Delta_{g},2\Delta_{g},\dots,D\}$ with $\Delta_{g}\ll\lambda/2$. The discretization error is negligible, since the acquisition function inherits from~\eqref{eq:kernel} a finest fringe of period $\lambda/2$.

\subsubsection{Complexity}
\label{subsubsec:complexity}

The complete procedure is summarized in Algorithm~\ref{alg:moveucb}. Its per-slot complexity is dominated by the posterior computation~\eqref{eq:posterior_mean}--\eqref{eq:posterior_var}: an $O(w^{3})$ factorization of the regularized Gram matrix, after which each grid point costs $O(w^{2})$, giving $O(w^{3}+Gw^{2})$ per slot for a grid of $G=\lceil D/\Delta_{g}\rceil+1$ points. Since the window slides by one sample per slot, the factorization admits a rank-one update and downdate, and the cubic term need not be recomputed at every slot. The movement-aware components add nothing to this cost: $\tilde{c}$ is a scalar computed once at initialization, and the persistence test in line~7 reuses the acquisition function already tabulated on the grid.

\section{Regret Analysis}
\label{sec:analysis}

We now analyze the regret of MoveUCB. The estimation and staleness terms are handled as in non-stationary kernelized bandits. The travel term is new: the persistence rule ties it to the number of target switches, and bounding that number requires the reverse direction of the information-gain decomposition.

\subsection{Preliminaries}
\label{subsec:prelim}

The analysis rests on two assumptions, one on the regularity of the reward and one on the observation noise.

\begin{assumption}[Reward regularity]
\label{as:reward}
For every $t$, the reward function $g_{t}$ belongs to the reproducing kernel Hilbert space $\mathcal{H}_{k}$ of the kernel~\eqref{eq:kernel}, with $\|g_{t}\|_{\mathcal{H}_{k}}\le B$ and $|g_{t}(x)|\le g_{\max}$ for all $x\in\mathcal{D}$.
\end{assumption}

\begin{assumption}[Observation noise]
\label{as:noise}
The noise sequence $\{\epsilon_{t}\}$ is conditionally $R$-sub-Gaussian given the history.
\end{assumption}

Assumption~\ref{as:reward} formalizes the spatial correlation, and $B$ is small only when the kernel places mass on the fringe band $[0,2/\lambda]$ of~\eqref{eq:power_profile}, as~\eqref{eq:kernel} does. Assumption~\ref{as:noise} covers the finite-sample estimation error of the reward. The small-scale fading is not part of $\epsilon_{t}$, since it is a deterministic function of position within a slot and is therefore carried by $g_{t}$. The reward sequence $\{g_{t}\}_{t=1}^{T}$ is treated as fixed but unknown rather than drawn from a distribution, and~\eqref{eq:regret} is a per-sequence guarantee, constrained only through the \emph{variation budget}~\cite{besbes2019,cheung2019} defined as
\begin{equation}
  V_{T}\triangleq\sum_{t=1}^{T-1}\big\|g_{t+1}-g_{t}\big\|_{\infty},
  \label{eq:variation}
\end{equation}
which accumulates the slot-to-slot drift without imposing any statistical drift model. The bound of this section is sublinear whenever $V_{T}=o(T)$, which covers the slowly time-varying regime in which the channel geometry changes gradually across slots while remaining persistent enough to be tracked.

A central quantity in kernelized bandit analysis~\cite{srinivas2010,chowdhury2017} is the maximum information gain from $n$ samples,
\begin{equation}
  \gamma_{n}\triangleq\max_{\{z_{1},\dots,z_{n}\}\subset\mathcal{D}}
  \tfrac{1}{2}\log\det\!\big(I+\varrho^{-1}
  [k(z_{i},z_{j})]_{i,j=1}^{n}\big).
  \label{eq:infogain}
\end{equation}
For the kernel~\eqref{eq:kernel}, the point mass at the origin contributes $O(\log n)$ and the $2Q+1$ Gaussian lobes are dominated up to a constant by a squared exponential of length scale $\min_q\varsigma_q$, so $\gamma_n=O((\log n)^2)$~\cite{srinivas2010}. It appears below at $n=w$, since the estimator retains at most $w$ samples at any time.


The analysis requires two sums of posterior uncertainty at sampled positions, one conditioning each position on those that precede it and the other on those that follow it. Both follow from the decomposition below. Since~\eqref{eq:infogain} is a determinant rather than a mutual information, the decomposition is algebraic and therefore applies to adaptively chosen positions. Its forward direction is standard~\cite{srinivas2010}, whereas the reverse direction is what Section~\ref{subsec:travel} requires.

\begin{lemma}[Chain decomposition of the information gain]
\label{lem:chain}
Let $x_{1},\dots,x_{n}\in\mathcal{D}$ be any positions, adaptively chosen or otherwise, and write $\sigma_{\mathcal{S}}(x)$ for the posterior standard deviation at $x$ given the positions indexed by $\mathcal{S}$, the empty conditioning being read as the prior. Then
\begin{align}
  &\sum_{i=1}^{n}\log\big(1+\varrho^{-1}\sigma_{1:i-1}^{2}(x_{i})\big)
  =\log\det\big(I+\varrho^{-1}K\big)\nonumber\\
  &=\sum_{i=1}^{n}\log\big(1+\varrho^{-1}\sigma_{i+1:n}^{2}(x_{i})\big)
  \;\le\;2\gamma_{n},
  \label{eq:chain}
\end{align}
where $K=[k(x_{i},x_{j})]_{i,j=1}^{n}$, and consequently
\begin{equation}
  \sum_{i=1}^{n}\sigma_{1:i-1}^{2}(x_{i})=O(\gamma_{n}),
  \qquad
  \sum_{i=1}^{n}\sigma_{i+1:n}^{2}(x_{i})=O(\gamma_{n}).
  \label{eq:chain-sum}
\end{equation}
\end{lemma}
\begin{IEEEproof}
Partition $\varrho I+K$ after its first $i-1$ rows and columns. The Schur complement of the leading block is $\varrho+\sigma_{1:i-1}^{2}(x_{i})$, since it coincides with the posterior variance~\eqref{eq:posterior_var} at $x_{i}$ given the preceding positions, so the block determinant formula gives $\det(\varrho I+K_{1:i})=\det(\varrho I+K_{1:i-1})(\varrho+\sigma_{1:i-1}^{2}(x_{i}))$.
Applying this recursively for $i=n,\dots,1$ and dividing by $\varrho^{n}$ yields the first equality of~\eqref{eq:chain}. The determinant is invariant under a permutation of the positions, which permutes the rows and columns of $K$ alike, so expanding in the reversed ordering yields the second. The inequality is~\eqref{eq:infogain}, which maximizes over sets and therefore covers every realized sequence. For~\eqref{eq:chain-sum}, every posterior variance lies in $[0,\sigma_{f}^{2}]$ and $\log(1+s)\ge s\log(1+c)/c$ on $[0,c]$ with $c\triangleq\varrho^{-1}\sigma_{f}^{2}$, so each sum of variances is at most $2\varrho\gamma_{n}c/\log(1+c)$.
\end{IEEEproof}

\vspace{0.1cm}
Finally, slot $t$ is a \emph{dwell} slot if $x_{t+1}=z_{t}$, so that the antenna is at its target, and a \emph{transit} slot otherwise, so that it is still en route. We write $T_{\mathrm{trans}}$ for the number of transit slots.

\subsection{Bounding the Travel}
\label{subsec:travel}

The actuation energy and the reward lost in transit both enter the regret in proportion to the distance covered, and both are therefore controlled by bounding the total travel. We bound it in two steps,
\begin{equation*}
  \sum_{t=1}^{T}|x_{t+1}-x_{t}|
  \;\underset{\text{Lemma~\ref{lem:travel}}}{\le}\;
  D\,(N_{\mathrm{sw}}+1)
  \;\underset{\text{Lemma~\ref{lem:switch}}}{\le}\;
  D\Big(2+\frac{2S_{T}}{\Delta}\Big),
\end{equation*}
where $N_{\mathrm{sw}}$ is the number of target switches and $S_{T}$ measures how fast the acquisition function moves. The travel is a property of the trajectory and admits no direct bound, whereas $S_{T}$ is a property of the posterior and is bounded by the kernel arguments of Appendix~\ref{app:surface}. The transit count is bounded by the same travel, since each transit slot displaces the antenna by exactly $\rho_{\max}$.

The acquisition function alone does not bound the travel. As observations expire, the posterior variance regrows at the positions they occupied, and a learner maximizing~\eqref{eq:acq-fn} afresh at every slot is drawn back to them indefinitely. Since $\tilde{c}$ vanishes with the horizon, the resulting exploration bonus eventually exceeds the movement cost for any distance, and the travel is $\Theta(\rho_{\max}T)$ even for a static reward function. What bounds it is the persistence rule~\eqref{eq:hysteresis}, which ties the movement to the number of target switches $N_{\mathrm{sw}}\triangleq|\{t\le T:z_{t}\neq z_{t-1}\}|$.

\begin{lemma}[Travel is paid per switch]
\label{lem:travel}
In every realization,
\begin{align}
  \sum_{t=1}^{T}|x_{t+1}-x_{t}|\;\le\;D\,(N_{\mathrm{sw}}+1),
  \nonumber\\
  T_{\mathrm{trans}}\;\le\;\frac{D\,(N_{\mathrm{sw}}+1)}{\rho_{\max}}.
  \label{eq:travel}
\end{align}
\end{lemma}
\begin{IEEEproof}
Partition the horizon at the switch slots into $N_{\mathrm{sw}}+1$ maximal runs on which the target is constant. On a run with target $z$, the rule~\eqref{eq:approach} moves $x_{t}$ monotonically toward $z$ and halts on arrival, so the distance covered on the run is at most $D$. Every transit slot displaces the antenna by exactly $\rho_{\max}$, so their number is the total travel divided by $\rho_{\max}$.
\end{IEEEproof}

A switch occurs only when the current target has fallen a margin $\Delta$ behind the best available position, and switches are therefore rare unless the acquisition function moves quickly. We measure its movement by $S_{T}\triangleq\sum_{t=1}^{T-1}\omega_{t}$ with $\omega_{t}\triangleq\sup_{x\in\mathcal{D}}|F_{t+1}(x)-F_{t}(x)|$.

\begin{lemma}[Switch count]
\label{lem:switch}
In every realization,
\begin{equation*}
    N_{\mathrm{sw}}\le 1+2S_{T}/\Delta.
\end{equation*}
\end{lemma}
\begin{IEEEproof}
Let $\Phi_{t}\triangleq\max_{x\in\mathcal{D}}F_{t}(x)$ and $G_{t}\triangleq\Phi_{t}-F_{t}(z_{t-1})\ge0$, so that~\eqref{eq:hysteresis} switches at slot $t$ if and only if $G_{t}>\Delta$. Whether or not a switch occurs at $t$, the identity $G_{t+1}-[\Phi_{t}-F_{t}(z_{t})]
=[\Phi_{t+1}-\Phi_{t}]-[F_{t+1}(z_{t})-F_{t}(z_{t})]$ holds, and its right-hand side is at most $2\omega_{t}$: the first bracket because a supremum is $1$-Lipschitz with respect to the uniform norm, the second pointwise. Let $t_{1}<\dots<t_{N_{\mathrm{sw}}}$ denote the switch slots and $\mathcal{I}_{j}\triangleq\{t_{j-1},\dots,t_{j}-1\}$. After a switch the target is the maximizer of $F_{t}$, so $\Phi_{t}-F_{t}(z_{t})$ vanishes at $t=t_{j-1}$; telescoping over $\mathcal{I}_{j}$ therefore gives $G_{t_{j}}\le2\sum_{t\in\mathcal{I}_{j}}\omega_{t}$, while a switch at $t_{j}$ requires $G_{t_{j}}>\Delta$. Each switch thus accounts for more than $\Delta/2$ of the surface movement, and the $\mathcal{I}_{j}$ being disjoint, summing over $j=2,\dots,N_{\mathrm{sw}}$ gives $(N_{\mathrm{sw}}-1)\Delta/2<\sum_{t=1}^{T-1}\omega_{t}=S_{T}$.
\end{IEEEproof}

It remains to bound $S_{T}$. The acquisition function moves for three reasons: the posterior changes as one observation enters the window and one leaves it, the reward function drifts, and the reference point $x_{t}$ of the movement cost term $\tilde{c}|x-x_{t}|$ moves. The last is $O(T/w)$ under~\eqref{eq:price} and is absorbed by the others, which Lemma~\ref{lem:surface} bounds as
\begin{equation}
  S_{T}=O\!\Big(\underbrace{a\,T\sqrt{\gamma_{w}/w}}
    _{\textnormal{forgetting}}
  +\underbrace{\frac{\sigma_{f}^{3}}{\varrho^{3/2}}\,w^{3/2}V_{T}}
    _{\textnormal{drift}}\Big)
  \label{eq:surface-bound}
\end{equation}
with probability at least $1-\delta/2$, where $a\triangleq\sigma_{f}\big(\varrho^{-1}(\beta\sigma_{f}+R) +\varrho^{-1/2}\beta\big)$. The forgetting term is what requires the reverse direction of Lemma~\ref{lem:chain}. Deleting the oldest sample restores the posterior variance at its position by an amount governed by the uncertainty that remains there once every later sample is accounted for, and that uncertainty admits the same bound as the forward direction supplies in the standard analysis.

The travel bound enters the regret through the actuation energy and the reward lost on transit slots. As detailed in Appendix~\ref{app:proof}, Assumption~\ref{as:reward} bounds the reward loss by $2g_{\max}$ per transit slot. Hence, by~\eqref{eq:travel}, the sum of the actuation energy and the transit reward loss is bounded by $C_0(N_{\mathrm{sw}}+1)$, where
\begin{equation}
  C_{0}\triangleq D\Big(c_{\mathrm{e}} +\frac{2g_{\max}}{\rho_{\max}}\Big).
  \label{eq:c0}
\end{equation}
Combining this bound with Lemma~\ref{lem:switch} yields the $O(C_0S_T/\Delta)$ term in the regret bound. A smaller $\Delta$ increases this contribution, which is the counterpart of the per-slot loss that a larger $\Delta$ incurs.

\subsection{Main Result}
\label{subsec:main}

The main result combines the travel bound of Section~\ref{subsec:travel} with the dwell-slot bound obtained from the supporting lemmas of Appendix~\ref{app:lemmas}.

\begin{theorem}[Regret of MoveUCB]
\label{thm:regret}
Suppose Assumptions~\ref{as:reward} and~\ref{as:noise} hold, and run Algorithm~\ref{alg:moveucb} with $\beta=B+\frac{R}{\sqrt{\varrho}}\sqrt{2(\gamma_{w}+1+\ln(2T/\delta))}$, grid spacing $\Delta_{g}=O(D/T)$, and margin $\Delta>0$. Then, with probability at least $1-\delta$,
\begin{equation}
  \mathcal{R}_{T}=O\Big(\underbrace{\beta T\sqrt{\gamma_{w}/w}}
    _{\textnormal{estimation}}
  +\underbrace{\frac{\sigma_{f}}{\sqrt{\varrho}}w^{3/2}V_{T}}
    _{\textnormal{staleness}}
  +\underbrace{\frac{C_{0}S_{T}}{\Delta}+\Delta T}
    _{\textnormal{movement and persistence}}\Big).
  \label{eq:regret-bound}
\end{equation}
\end{theorem}
\begin{IEEEproof}
See Appendix~\ref{app:proof}.
\end{IEEEproof}

The bound is established against the per-slot maximizer $\max_{x\in\mathcal{D}}g_{t}(x)$ rather than against the oracle~\eqref{eq:oracle}: since the movement cost is nonnegative and $\mathcal{A}(x_{t-1})\subseteq\mathcal{D}$, the oracle is upper-bounded by $\sum_{t}\max_{x\in\mathcal{D}}g_{t}(x)$. The comparison is possible even though no policy can follow the maximizer within the speed limit, because the slots on which the antenna is not at its target are counted in full, and Lemma~\ref{lem:switch} keeps their number sublinear whenever the variation budget is.

The four terms fall into two pairs, each representing a trade-off in one parameter. The first pair is that of a non-stationary kernelized bandit: a longer window $w$ averages more noise but retains staler observations. The second pair is the trade-off of Section~\ref{subsec:travel}: a smaller margin $\Delta$ keeps the target closer to the maximizer of $F_{t}$, at the expense of more frequent switching and hence more travel. The movement cost coefficient $\tilde{c}$ does not appear at all. Its residual, $\tilde{c}D$ per slot, sums to $O(T/w)$ under~\eqref{eq:price} and is absorbed by the estimation term, and the guarantee therefore holds with or without it. What $\tilde{c}$ governs is the distance of a switch rather than its frequency, a distinction that Lemma~\ref{lem:travel} does not resolve, since it counts every switch as a full traversal. Section~\ref{subsec:results} reports its effect over a finite horizon.

The rate follows by balancing the two pairs in turn. The margin appears only in the second pair, and equating its two terms gives
\begin{equation}
  \Delta = \Theta\big(\sqrt{C_{0}\bar{S}_{T}/T}\big),
  \label{eq:margin}
\end{equation}
where $\bar{S}_{T}$ denotes the bound~\eqref{eq:surface-bound}, which holds for every trajectory and is therefore independent of $\Delta$. Substituting~\eqref{eq:surface-bound} into~\eqref{eq:margin} splits the movement contribution $2\sqrt{C_{0}\bar{S}_{T}T}$ into a part that decreases in $w$ and a part that increases in it,
\begin{equation*}
  O\Big(\sqrt{C_{0}a}\,T(\gamma_{w}/w)^{1/4}\Big)
  +O\Big(\sqrt{C_{0}\sigma_{f}^{3}\varrho^{-3/2}}\,
  T^{1/2}w^{3/4}V_{T}^{1/2}\Big),
\end{equation*} and balancing the two gives $w=\tilde{\Theta}\big((T/V_{T})^{1/2}\big)$. At this window the estimation and staleness terms are both $\tilde{O}(T^{3/4}V_{T}^{1/4})$ and are dominated, and \eqref{eq:regret-bound} thus reduces to
\begin{equation}
  \mathcal{R}_{T}=\tilde{O}\big(T^{7/8}V_{T}^{1/8}\big),
  \label{eq:rate}
\end{equation}
where $\tilde{O}(\cdot)$ hides polylogarithmic factors in $T$. The regret is therefore \emph{sublinear} in $T$ whenever $V_{T}=o(T)$, with no further condition.

Two features of~\eqref{eq:rate} are worth noting. The exponent $7/8$ is set by the movement term, whereas the estimation and staleness terms retain the exponent $3/4$ of non-stationary kernelized bandits; the gap is the cost of a decision variable that must travel to where it is evaluated. The physical parameters $D$, $\rho_{\max}$, and $c_{\mathrm{e}}$ enter only through $C_{0}$, so a slower or more expensive actuator raises the constant but not the exponent.

\section{Simulation Results}
\label{sec:sim}

This section validates the theoretical analysis and evaluates MoveUCB against the movement-aware oracle and several baselines.

\subsection{Benchmarks}
\label{subsec:benchmarks}

MoveUCB is compared against the following baselines, all of which are charged the same movement cost $c_{\mathrm{e}}\sum_{t}|x_{t}-x_{t-1}|$ and subject to the same reachability constraint~\eqref{eq:reachable}.

\noindent\textbf{Oracle} computes~\eqref{eq:oracle} with non-causal knowledge of $\{g_{t}\}_{t=1}^{T}$. The objective and the constraint couple only consecutive positions, so it is solved exactly by backward dynamic programming on the grid $\mathcal{X}_{g}$. It defines the regret rather than competing with the other schemes.

\noindent\textbf{GP-UCB}~\cite{srinivas2010,chowdhury2017} maintains the same windowed estimator~\eqref{eq:posterior_mean} and~\eqref{eq:posterior_var} as MoveUCB, with the same kernel, window, and confidence parameter, and selects $x_{t+1}=\arg\max_{x\in\mathcal{A}(x_{t})\cap\mathcal{X}_{g}} [\mu_{t}(x)+\beta\sigma_{t}(x)]$. Since the maximization is confined to the reachable interval, no target beyond one slot of travel is ever evaluated, and the comparison with MoveUCB therefore isolates the effect of the two-stage selection.

\noindent\textbf{Sweep-and-hold} is a measure-then-optimize benchmark. Its first cycle coincides with the calibration described in Section~\ref{subsec:setup} and consumes none of the evaluated slots. In every subsequent cycle, of period $P$, it sweeps at full stroke for $T_{0}$ slots, reflecting at the boundaries of the region and recording the reward at each position visited, then returns to the position with the highest recorded value and holds there until the cycle restarts; unlike this initial calibration, these sweeps consume $T_{0}$ evaluated slots each and are charged in full. The period is selected from $P\in\{500,1000,2000\}$ in the partially blocked scenario and applied throughout.

\noindent\textbf{$\epsilon$-greedy} maintains the same estimator as MoveUCB and moves to the maximizer of the posterior mean within reach with probability $1-\epsilon$, $\epsilon=0.1$, and to a uniformly random reachable position otherwise. Its exploration is undirected, being guided by neither the posterior uncertainty nor the movement cost.

\noindent\textbf{Fixed} holds the initial position throughout, representing a conventional antenna. It incurs no movement cost and requires no estimation, so its regret measures the value of repositioning itself.

\subsection{Setup}
\label{subsec:setup}

All parameters are listed in Table~\ref{tab:params}; this subsection states only what they are chosen for. With the region length and stroke of Table~\ref{tab:params}, a traversal takes $100$ slots, or $0.5$~s, and the horizon of $T=15300$ slots spans $76.5$~s. The reward is the rate~\eqref{eq:reward} at a transmit SNR of $10$~dB, and its maximum $g_{\max}\approx3.96$ is attained when all paths add coherently. The movement cost coefficient is set so that a traversal costs $c_{\mathrm{e}}D=1.8$, less than half of $g_{\max}$; the energy is thus a secondary component of the movement cost, as anticipated in Section~\ref{subsubsec:price}.

Every policy except the fixed antenna is calibrated from $T_{0}=100$ samples of the reward map at $t=0$: spatial samples of a single frozen channel snapshot, not measurements collected over a sequence of slots, so no reachability constraint applies to them. This models a site survey obtained before deployment (by ray-tracing or a prior measurement campaign, say) rather than real-time operation of the movable antenna, and consumes none of the $T$ evaluated slots. For the GP-based schemes---MoveUCB, GP-UCB, and $\epsilon$-greedy---the samples fill the window and calibrate $\sigma_{f}^{2}$; their evaluated trajectory, subject to the reachability constraint~\eqref{eq:reachable} throughout, begins at $x_{0}=D/2$ independently of the survey. For Sweep-and-hold, which has no estimator to calibrate, the same samples instead select its first holding position, so its evaluated trajectory begins already held there. The fixed antenna requires no such survey. All results are averaged over $48$ independent realizations of the channel, the blockage position, and the noise, and the error bars denote $\pm$ one standard error.

\begin{table}[t]
\centering
\caption{Simulation parameters.}
\label{tab:params}
\renewcommand{\arraystretch}{1.05}
\setlength{\tabcolsep}{3.5pt}
\resizebox{\columnwidth}{!}{%
\begin{tabular}{@{}llc@{\hspace{8pt}}lc@{}}
\toprule
& Parameter & Value & Parameter & Value \\
\midrule
\multirow{4}{*}{\rotatebox{90}{Phys.}}
& Carrier frequency $f_{c}$ & $30$~GHz
& Stroke $\rho_{\max}=vT_{\mathrm{a}}$ & $0.06\lambda$ \\
& Region length $D$ & $6\lambda$
& Movement cost $c_{\mathrm{e}}$ & $0.3$ per $\lambda$ \\
& Slot / actuation $T_{\mathrm{s}},T_{\mathrm{a}}$ & $5$, $1$~ms
& Initial position $x_{0}$ & $D/2$ \\
& Actuator speed $v$ & $0.6$~m/s
& Horizon $T$ / warm-up $T_{0}$ & $15300$ / $100$ \\
\cmidrule(r){1-5}
\multirow{4}{*}{\rotatebox{90}{Channel}}
& Paths $L$ (LoS $+$ scattered) & $3$ $(1+2)$
& Blockage $A_{\mathrm{b}}$ (blk. / dense) & $10$ / $0$~dB \\
& Path powers & $(0.6,0.2,0.2)$
& Shadow extent / edge & $D/2$, $0.3\lambda$ \\
& Angular sector & $[30^{\circ},150^{\circ}]$
& Transmit SNR & $10$~dB \\
& Angular spread / sub-rays & $2^{\circ}$ / $16$
& Noise level $\sigma_{\epsilon}$ & $0.05\,g_{\max}$ \\
\cmidrule(r){1-5}
\multirow{2}{*}{\rotatebox{90}{Drift}}
& Angular $(\rho_{\theta},\sigma_{\theta})$ & $(0.9999,\,0.1^{\circ})$
& Shadow drift $\sigma_{\mathrm{s}}$ & $0.01\lambda$ per slot \\
& Gain correlation $\rho_{\alpha}$ & $0.9997$
& Change-point $\sigma_{\theta}^{\mathrm{cp}}$ & $0.5^{\circ}$ \\
\cmidrule(r){1-5}
\multirow{5}{*}{\rotatebox{90}{Algorithm}}
& Kernel components $Q$ & $3$
& Window length $w$ & $300$ \\
& Fringe periods $\{\lambda_{q}\}$ & $\{4.15,1.81,1.04\}\lambda$
& Movement cost weight $\eta$ & $2$ \\
& Damping $\varsigma_{q}$, $\varsigma_{0}$ & $2\lambda$, $4\lambda$
& Persistence margin $\Delta$ & $0.1\,\sigma_{f}$ \\
& Weights $\sigma_{c}^{2},\sigma_{0}^{2},\sigma_{q}^{2}$
  & $\sigma_{f}^{2}\{\tfrac{1}{4},\tfrac{1}{4},\tfrac{1}{2Q}\}$
& ridge $\varrho$ & $0.1$ \\
& & & Grid spacing $\Delta_{g}$ & $\lambda/100$ \\
\bottomrule
\end{tabular}}
\end{table}

\subsection{Scenarios and Drift Models}
\label{subsec:scenarios}

Two axes are varied. The scenario determines how far a local search can fall short of the global optimum, and hence the potential gain from a global target search, whereas the drift model determines whether the variation budget satisfies the sublinearity condition $V_{T}=o(T)$ of Theorem~\ref{thm:regret}.

\noindent\emph{Scenarios.} In the \emph{dense-multipath} scenario the channel follows the multipath model~\eqref{eq:multipath} without obstruction, and the gap to the global optimum comes from multipath interference alone, the fringes of~\eqref{eq:power_profile} being sparse at $L=3$. In the \emph{partially blocked} scenario an obstacle attenuates the line-of-sight (LoS) path over a contiguous sub-interval $[s_{t},s_{t}+D/2]$, with soft edges that model diffraction around it, so that a search started in the shadow has no reachable direction of improvement and the gap widens further.

\noindent\emph{Drift models.} The \emph{continuous} model represents a persistently moving environment, such as a walking user or a scatterer in steady motion. The angles, gains, and obstacle position evolve at every slot: the angles by $\bar{\theta}_{\ell}^{(t+1)}=\rho_{\theta}\bar{\theta}_{\ell}^{(t)}+(1-\rho_{\theta})\bar{\theta}_{\ell}^{(0)}+\zeta_{\ell}^{(t)}$ with $\zeta_{\ell}^{(t)}\sim\mathcal{N}(0,\sigma_{\theta}^{2})$, the gains by a first-order Gauss--Markov process that preserves the path powers, and the obstacle position by a mean-reverting recursion. The rates are set so that the per-slot maximizer moves more slowly than the actuator, since a channel whose optimum outruns it cannot be tracked by any policy; at the values of Table~\ref{tab:params} the maximizer moves by $0.010\lambda$ per slot against a stroke of $0.06\lambda$. The drift process is statistically stationary and never settles, so the variation accumulates at a constant rate and $V_{T}=\Theta(T)$. The \emph{piecewise-stationary} model represents an environment reconfigured by discrete events rather than by continuous motion---a blockage clearing, a scatterer entering, the user reorienting. Such events are how the dominant paths change at millimeter-wave frequencies, and this is also the regime in which past observations retain lasting value. The geometry is held fixed except at $S$ change points placed uniformly at random, at each of which the angles are perturbed, the gains are partially refreshed, and the obstacle relocates. This gives $V_{T}=\Theta(S)$, so that sweeping $S$ varies the budget at a fixed horizon, and letting $S$ grow more slowly than $T$ realizes the regime $V_{T}=o(T)$.

\subsection{Results}
\label{subsec:results}

We first compare MoveUCB against the baselines, then examine the variation budget, the kernel, and, in an ablation study, the movement-aware components and the selection rule.

\noindent\emph{Comparison against the baselines.} Fig.~\ref{fig:main} reports the time-averaged regret and the total travel under the continuous drift model. MoveUCB attains the lowest regret in both scenarios with the least travel of any policy that moves, about half that of GP-UCB. The gap to GP-UCB is attributable to the movement-aware selection rule, since the two policies share the estimator, the kernel, and the window, each running at its own optimal confidence parameter of Fig.~\ref{fig:beta}, $\beta=1.5$ for MoveUCB and $\beta=3.5$ for GP-UCB. The rule lowers the regret by more than a third under blockage, where the reward is nearly flat inside the shadow and no position within one slot of travel is better than the current one: GP-UCB must explore aggressively to walk out, whereas MoveUCB selects a target outside the shadow and reaches it over several slots, target persistence holding that target fixed until it arrives. Once the obstacle is removed the gap narrows, indicating that part of it stems from the reachability of the optimum. Sweep-and-hold also evaluates positions beyond the reachable interval and therefore escapes the shadow, but it commits on single noisy observations, does not update during the hold, and sweeps the region again at every cycle. Target persistence gives MoveUCB the opposite behavior: under blockage it commits to $N_{\mathrm{sw}}=298$ journeys over the horizon and dwells in between, spending $10\%$ of the slots in transit, which is the behavior that Lemma~\ref{lem:travel} assumes and Lemma~\ref{lem:switch} bounds.

\begin{figure}[t]
\centering
\includegraphics[width=0.90\columnwidth]{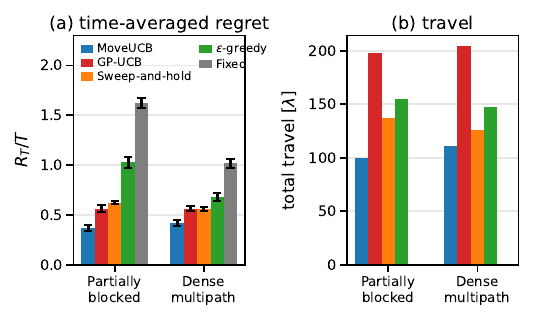}
\caption{Comparison against the baselines under the continuous drift model, in both scenarios, at $V_{T}/T\approx0.10$.}
\label{fig:main}
\end{figure}

\begin{figure}[t]
\centering
\includegraphics[width=0.90\columnwidth]{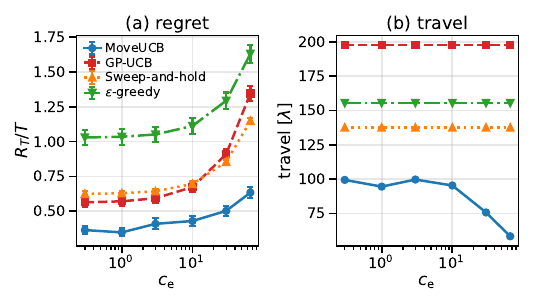}
\caption{Sensitivity to the movement cost coefficient $c_{\mathrm{e}}$,
partially blocked scenario, continuous drift model.}
\label{fig:ce}
\end{figure}

\noindent\emph{Sensitivity to the movement cost.} Fig.~\ref{fig:ce} examines whether the advantage of Fig.~\ref{fig:main} is an artifact of a small energy penalty, in the partially blocked scenario where that advantage is largest and therefore the more demanding test. At the calibrated $c_{\mathrm{e}}=0.3$, a lower travel barely registers in the regret, since the energy it saves is a small fraction of $g_{\max}$; Fig.~\ref{fig:ce}(a) shows what happens once the energy penalty is no longer negligible. The movement-unaware baselines travel a fixed distance regardless of $c_{\mathrm{e}}$, so their regret absorbs the energy penalty in full, whereas MoveUCB's movement cost term~\eqref{eq:price} lets it travel less as $c_{\mathrm{e}}$ grows---from $99\lambda$ at the default to $58\lambda$ at $c_{\mathrm{e}}=66$, the boundary of Remark~\ref{rem:budget}---so its margin over the closest competitor widens rather than narrows, from $0.20$ to $0.51$. Which baseline is closest also changes: GP-UCB, whose travel is largest, gives way to Sweep-and-hold once its own energy penalty catches up, between $c_{\mathrm{e}}=10$ and $30$. The ranking among the movement-unaware baselines is decided by travel alone once $c_{\mathrm{e}}$ dominates, and MoveUCB remains ahead of whichever of them travels least.

\noindent\emph{Dependence on the variation budget.} Fig.~\ref{fig:vt}(a) separates the two drift models by the budget they accumulate. The continuous model is stationary and its budget grows with a fitted exponent of $1.02$ on $T$, so the condition $V_{T}=o(T)$ does not hold, whereas the piecewise-stationary model with $S=\lceil T^{1/2}\rceil$ change points gives $0.48$ and the condition holds. Fig.~\ref{fig:vt}(b) sweeps the budget at a fixed horizon under the piecewise model, from a geometry reconfigured every $3060$ slots to one reconfigured every $51$, half the time a traversal takes. MoveUCB leads at every point, but which baseline comes closest depends on the budget, and the two exchange order because they fail for opposite reasons. Sweep-and-hold commits once per cycle, which works when a good position lasts for thousands of slots but fails when the held position degrades before the next sweep. GP-UCB re-explores its neighborhood at every slot, which costs it when the geometry is nearly static but makes it the least sensitive of the three to the budget. MoveUCB avoids both limitations. Its regret still grows with the budget, because crossing the region takes $D/\rho_{\max}=100$ slots and that cost pays off only while the chosen target remains good. Its advantage is therefore largest when the channel is reconfigured rarely relative to the time a relocation takes, precisely the regime for which a mechanical actuator is well suited.

\begin{figure}[t]
\centering
\includegraphics[width=0.95\columnwidth]{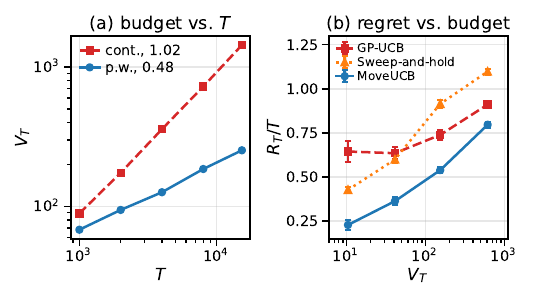}
\caption{Variation budget and its effect on the regret, partially blocked scenario.}
\label{fig:vt}
\end{figure}

\begin{figure}[t]
\centering
\includegraphics[width=0.90\columnwidth]{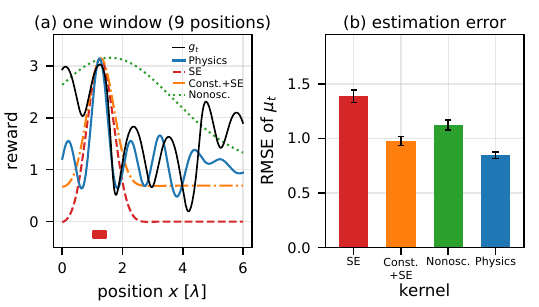}
\caption{Kernel comparison, partially blocked scenario, continuous drift model.}
\label{fig:kernel_abalation}
\end{figure}

\vspace{2pt}
\noindent\emph{Kernel.} The kernel enters the target decision through both the estimate $\mu_{t}$ and its uncertainty $\sigma_{t}$, and Fig.~\ref{fig:kernel_abalation} examines its contribution by measuring the root mean square error (RMSE) of $\mu_{t}$ over the region under the sample distribution MoveUCB itself produces. The comparison is against the squared exponential (SE) $k_{\mathrm{SE}}(x,x')=\sigma_{f}^{2}\exp(-(x-x')^{2}/(2\varsigma^{2}))$, the default choice in kernelized bandits, a constant-plus-SE kernel (Const.+SE) given by $\sigma_{f}^{2}/4+3k_{\mathrm{SE}}/4$, and a nonoscillatory variant (Nonosc.) that replaces each cosine in~\eqref{eq:kernel} with one while retaining the weights and damping lengths; all four satisfy $k(x,x)=\sigma_{f}^{2}$. The SE and Const.+SE
correlation lengths are selected separately from $\varsigma\in\{0.2,0.5,1,2,3,4\}\lambda$ on eight validation realizations, both selecting $0.5\lambda$. The proposed kernel (Physics) is given no such freedom, its hyperparameters being fixed offline from the angular sector and the angular spread. The window in Fig.~\ref{fig:kernel_abalation}(a) holds $300$ observations but covers only nine distinct positions, since the antenna dwells at its target most of the time. The estimate over the rest of the region is an extrapolation rather than an interpolation. All four kernels are evaluated on the same observation windows collected by MoveUCB, so that the comparison in Fig.~\ref{fig:kernel_abalation}(b) turns on the kernel alone, and the proposed kernel attains the lowest RMSE. The extrapolation over the unvisited part of the region requires representing the fringes away from a sample. Neither the SE nor the constant-plus-SE kernel does so, and the nonoscillatory variant loses them as well.

\begin{figure}[t]
\centering
\includegraphics[width=0.92\columnwidth]{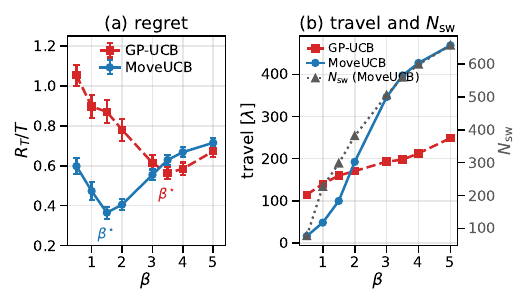}
\caption{Effect of the confidence parameter, partially blocked scenario, continuous drift model.}
\label{fig:beta}
\end{figure}

\vspace{2pt}
\noindent\emph{Confidence parameter.} 
Fig.~\ref{fig:beta} sweeps $\beta$ over a range that contains the optimum of both policies. MoveUCB attains its lowest regret at $\beta=1.5$, which is adopted as the default. Beyond this point MoveUCB degrades as $\beta$ grows, its travel increasing by an order of magnitude and the fraction of slots spent in transit approaching half the horizon. This is the mechanism of Section~\ref{subsec:travel} made visible: a larger $\beta$ raises the bonus $\beta\sigma_{t}$ that forgetting restores at every position the antenna has left, and once that bonus exceeds the margin $\Delta$ the persistence rule no longer holds the target, so the antenna spends the horizon traveling rather than at a position it selected. GP-UCB, which does not commit to a distant target, improves up to $\beta=3.5$, since aggressive exploration is how a policy confined to the reachable interval eventually escapes the shadow; its regret reaches a shallow minimum there and increases significantly at $\beta=5$. Even at its own optimum, GP-UCB remains $55\%$ above MoveUCB and travels twice as far. Each policy runs at its own optimum, $\beta=1.5$ for MoveUCB and $\beta=3.5$ for GP-UCB, in Figs.~\ref{fig:main} and~\ref{fig:vt}.

\noindent\emph{Ablations.} Fig.~\ref{fig:ablation} varies the two movement-aware components in turn, and they affect the regret differently. The weight $\eta$ leaves the regret unchanged within about one standard error on paired comparison, while the travel decreases monotonically from the case $\eta=0$, in which the movement cost term~\eqref{eq:price} reduces to the actuation energy alone. The term therefore reduces the travel rather than the regret, which is consistent with its role in the analysis: its residual is absorbed by the estimation term of Theorem~\ref{thm:regret}, so the guarantee holds with or without it, and it controls how far each switch travels rather than how often switches occur. Accordingly, the switch count varies little over the sweep, whereas $\Delta$ reduces it by nearly two orders of magnitude. We adopt $\eta=2$, which yields the lowest mean regret and reduces the travel by $40\%$ relative to $\eta=0$.

\begin{figure}[t]
\centering
\includegraphics[width=0.88\columnwidth]{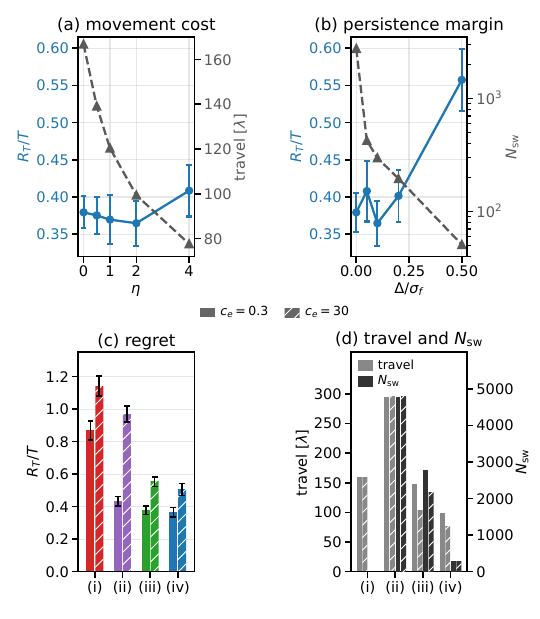}
\caption{Ablation of the two movement-aware components and of the selection rule, partially blocked scenario, continuous drift model: (i) local (GP-UCB's rule); (ii) global; (iii) global + cost; (iv) global + cost + persistence (MoveUCB). (a)--(b) use $c_{\mathrm{e}}=0.3$; (c)--(d) compare $c_{\mathrm{e}}=0.3$ (solid) and $c_{\mathrm{e}}=30$ (hatched). $N_{\mathrm{sw}}$ is undefined for (i) (no persisted target).}
\label{fig:ablation}
\end{figure}

The margin $\Delta$ introduces a trade-off between travel and regret. As $\Delta/\sigma_{f}$ increases from $0$ to $0.5$, the switch count decreases from $2770$ to $51$ and the travel with it, both monotonically, while the regret remains unchanged up to $\Delta=0.2\,\sigma_{f}$ and increases significantly beyond, being
higher at $\Delta=0.5\,\sigma_{f}$ on $42$ of $48$ realizations. This is the trade-off predicted by~\eqref{eq:margin}: the switch count decreases as $1/\Delta$, whereas the per-slot loss increases as $\Delta$. We adopt $\Delta=0.1\,\sigma_{f}$, which yields the lowest mean regret within the range where the travel has already been reduced and the regret has not yet increased.

Figs.~\ref{fig:ablation}(c)--(d) isolate the contribution of the selection rule by starting from GP-UCB's rule and adding one component at a time. (i) is GP-UCB's own rule, confined to the reachable interval. (ii) adds global search: the maximization ranges over the whole region, so the antenna sets out for the best position it sees but is drawn back to wherever the estimate has gone stale, since neither the movement cost term nor persistence is yet present to prevent it. (iii) further adds the movement cost term~\eqref{eq:price}. (iv), MoveUCB, further adds target persistence~\eqref{eq:hysteresis}. At $c_{\mathrm{e}}=0.3$, the largest gain is from (i) to (ii), which nearly halves the regret: this isolates the shadow-escape effect of Fig.~\ref{fig:main}, since neither the cost term nor persistence is present yet to explain it. From (ii) onward the regret changes little, while the travel and the switch count fall in two stages: the cost term alone roughly halves both, from $N_{\mathrm{sw}}=4778$ at (ii) to $2770$ at (iii), and persistence then reduces the count by a further order of magnitude, to $298$ at (iv).

Neither (i) nor (ii) prices the distance, so their travel at $c_{\mathrm{e}}=30$ is identical to that at $c_{\mathrm{e}}=0.3$ ($N_{\mathrm{sw}}$ likewise for (ii); it is undefined for (i)): the two policies cannot respond to $c_{\mathrm{e}}$ at all. This is where the movement cost term earns its keep: consistent with Fig.~\ref{fig:ce}, raising $c_{\mathrm{e}}$ to $30$ turns it from an efficiency device into a regret-relevant one, since (iii) and (iv) can now trade travel for regret while (i) and (ii) cannot. The mean regret reduction from (ii) to (iii) grows from $0.054$ at $c_{\mathrm{e}}=0.3$ to $0.417$ at $c_{\mathrm{e}}=30$, and the travel of (iii) and (iv) falls accordingly. What Section~\ref{subsec:main} anticipated as a term that only buys distance without a measurable cost in regret is therefore conditional on $c_{\mathrm{e}}$ being small: as it grows, the same term becomes what keeps the search accurate as well as affordable.

\section{Conclusion}\label{sec:conclusion}

This paper studied movable antenna position learning without channel knowledge, where a position is evaluated only by moving there and the channel drifts in the meantime. We formulated the problem as a non-stationary kernelized bandit under a reachability constraint and proposed MoveUCB, which combines a physics-informed kernel with a global target search, a persistence rule, and a transit-aware movement cost term. The regret is sublinear whenever the channel variation is sublinear, with no further condition, and the physical parameters enter only through constants. Simulations show that accounting for movement lowers the regret and the travel together: MoveUCB attains the lowest regret among the baselines while covering about half the travel of the closest one. Two directions remain open: arrays of multiple movable antennas, where the reward couples all positions and a selection rule must arbitrate among antennas as well as positions; and a displacement-dependent actuation interval, under which the cost of a relocation is coupled to the reward being estimated and the optimism argument must be revised.

\appendices
\section{Proof of Theorem~\ref{thm:regret}}
\label{app:proof}

Write $x_{t}^{\star}\in\arg\max_{x\in\mathcal{D}}g_{t}(x)$ for the comparator. Since the movement cost in~\eqref{eq:oracle} is nonnegative and $\mathcal{A}(x_{t-1})\subseteq\mathcal{D}$, we have $\mathrm{OPT}_{T}\le\sum_{t}g_{t}(x_{t}^{\star})$, so we need not reason about the oracle trajectory: the optimism argument of Lemma~\ref{lem:dwell} holds against any fixed position in $\mathcal{D}$, and $x_{t}^{\star}$ is admissible whether or not the learner could reach it. We work on the intersection of the confidence event of Lemma~\ref{lem:confidence} with the noise event $E_{\epsilon}$ of Lemma~\ref{lem:surface}, which has probability at least $1-\delta$ by the union bound. Substituting into~\eqref{eq:regret},
\begin{equation}
  \mathcal{R}_{T}\le
  \underbrace{\sum_{t=1}^{T}
  \mathbb{E}\big[g_{t}(x_{t}^{\star})-g_{t}(x_{t})\big]}_{%
  \triangleq\,\mathcal{R}^{\mathrm{val}}_{T}}
  +\underbrace{c_{\mathrm{e}}\sum_{t=1}^{T}
  \mathbb{E}\big|x_{t}-x_{t-1}\big|}_{%
  \triangleq\,\mathcal{R}^{\mathrm{mov}}_{T}}.
  \label{eq:decomposition}
\end{equation}

\noindent\emph{Index alignment.} The position $x_{t+1}$ is selected under $g_{t}$ but earns its reward under $g_{t+1}$, and Lemma~\ref{lem:dwell} accordingly controls $g_{t}(x_{t}^{\star})-g_{t}(x_{t+1})$, whereas $\mathcal{R}^{\mathrm{val}}_{T}$ pairs $g_{t}$ with $x_{t}$. Since $g_{t}(x_{t})\ge g_{t-1}(x_{t})-\|g_{t}-g_{t-1}\|_{\infty}$ for $t\ge2$, summing over $t=2,\dots,T$ and bounding the unmatched first slot by $g_{\max}$ gives $\sum_{t}g_{t}(x_{t})\ge\sum_{t\le T-1}g_{t}(x_{t+1})-V_{T}-g_{\max}$. Applying this to the learner's sum and isolating $t=T$,
\begin{equation}
  \mathcal{R}^{\mathrm{val}}_{T}
  \le\sum_{t=1}^{T-1}
  \mathbb{E}\big[g_{t}(x_{t}^{\star})-g_{t}(x_{t+1})\big]
  +V_{T}+2g_{\max}.
  \label{eq:shift}
\end{equation}

\noindent\emph{Slot-wise accounting.} On a dwell slot,
Lemma~\ref{lem:dwell} with $\tilde{x}=x_{t}^{\star}$ bounds the summand
of~\eqref{eq:shift} by
$2\beta\sigma_{t}(x_{t+1})+2S_{t}+\tilde{c}D+\Delta$; on a transit slot
it is at most $2g_{\max}$ by Assumption~\ref{as:reward}. Extending the
dwell sum to all slots, which is valid as every summand is nonnegative,
and applying Lemma~\ref{lem:optsum} to the first term and
Lemma~\ref{lem:confidence} to the second,
\begin{equation}
  \mathcal{R}^{\mathrm{val}}_{T}
  =O\Big(\beta T\sqrt{\gamma_{w}/w}
  +\frac{\sigma_{f}}{\sqrt{\varrho}}\,w^{3/2}V_{T}
  +\Delta T
  +g_{\max}T_{\mathrm{trans}}\Big),
  \label{eq:val-bound}
\end{equation}
the $V_{T}$ and $2g_{\max}$ of~\eqref{eq:shift} being absorbed into the staleness term and the constant, and the movement cost residual $\tilde{c}DT=O(T/w)$ by~\eqref{eq:price} being dominated by the estimation term $\Theta(T/\sqrt{w})$.

\noindent\emph{Combination.} Lemma~\ref{lem:travel} bounds the movement expenditure and the transit count by the same travel, so $\mathcal{R}^{\mathrm{mov}}_{T}+2g_{\max}T_{\mathrm{trans}}
\le(c_{\mathrm{e}}+2g_{\max}/\rho_{\max})D(N_{\mathrm{sw}}+1)=C_{0}(N_{\mathrm{sw}}+1)$ with $C_{0}$ as in~\eqref{eq:c0}, and Lemma~\ref{lem:switch} gives $N_{\mathrm{sw}}+1\le2+2S_{T}/\Delta$. Substituting this and~\eqref{eq:val-bound} into~\eqref{eq:decomposition} gives~\eqref{eq:regret-bound} up to the grid discretization, the additive constant $C_{0}$ being absorbed.

\noindent\emph{Discretization.} Lemma~\ref{lem:dwell} is stated for the exact maximizer of $F_{t}$ over $\mathcal{D}$, whereas Algorithm~\ref{alg:moveucb} maximizes over $\mathcal{X}_{g}$. The per-slot loss is at most $L_{U}\Delta_{g}$ with $L_{U}=O(L_{k}(B+R\sqrt{\gamma_{w}/\varrho}+\beta)+\tilde{c})$ and $L_{k}\triangleq\sup_{r\ne0}\sqrt{2(k(0)-k(r))}/|r|<\infty$ by the twice differentiability of~\eqref{eq:kernel} at the origin, so the total grid error is $O(L_{U}\Delta_{g}T)=O(1)$ under $\Delta_{g}=O(D/T)$. Lemma~\ref{lem:switch} is unaffected, its invariant being stated relative to the maximum actually computed.\hfill$\blacksquare$

\section{Proof of Supporting Lemmas}
\label{app:lemmas}

\subsection{Confidence Bound Under Drift}
\label{app:confidence}

The estimator of Section~\ref{subsec:kernel} faces two error sources: observation noise, and staleness, as the window retains samples of earlier functions $g_{\tau}$ that differ from the current $g_{t}$. The staleness term is bounded in the Euclidean norm of the drift vector rather than in its $\ell_{1}$ norm, which is what yields the $w^{3/2}V_{T}$ scaling in place of $w^{2}V_{T}$ and hence the exponent $3/4$ of the final rate. For the window $\mathcal{W}_{t}$ we write $V_{\mathcal{W}_{t}}\triangleq\sum_{\tau\in\mathcal{W}_{t},\,\tau<t}\|g_{\tau+1}-g_{\tau}\|_{\infty}$ for the drift accumulated over it, so that $V_{\mathcal{W}_{t}}\le V_{T}$ and, since each single-step variation belongs to at most $w$ windows, $\sum_{t=1}^{T}V_{\mathcal{W}_{t}}\le wV_{T}$.

\begin{lemma}[Time-varying confidence bound]
\label{lem:confidence}
Let Assumptions~\ref{as:reward} and \ref{as:noise} hold and let $\beta$
be as in Theorem~\ref{thm:regret}. Then, with probability at least
$1-\delta/2$, for all $t\ge1$ and all $x\in\mathcal{D}$,
\begin{equation}
  \big|g_{t}(x)-\mu_{t}(x)\big|\le\beta\sigma_{t}(x)+S_{t},
  \qquad
  S_{t}\triangleq\frac{\sigma_{f}}{\sqrt{\varrho}}\sqrt{w}\,
  V_{\mathcal{W}_{t}}.
  \label{eq:conf}
\end{equation}
Consequently,
\begin{equation}
  \sum_{t=1}^{T}S_{t}\le\frac{\sigma_{f}}{\sqrt{\varrho}}\,w^{3/2}V_{T}.
  \label{eq:staleness-sum}
\end{equation}
\end{lemma}
\begin{IEEEproof}
Express each windowed observation with respect to the current function, $r_{\tau}=g_{t}(x_{\tau})+d_{\tau}+\epsilon_{\tau}$ with $d_{\tau}\triangleq g_{\tau}(x_{\tau})-g_{t}(x_{\tau})$, and let $\mathbf{M}_{t}\triangleq\mathbf{K}_{t}+\varrho\mathbf{I}$, $\mathbf{d}_{t}\triangleq[d_{\tau}]_{\tau\in\mathcal{W}_{t}}$, $\mathbf{g}_{t}\triangleq[g_{t}(x_{\tau})]_{\tau\in\mathcal{W}_{t}}$,
and $\boldsymbol{\epsilon}_{t}\triangleq[\epsilon_{\tau}]_{\tau\in\mathcal{W}_{t}}$.
By the linearity of \eqref{eq:posterior_mean} in the observation vector,
\begin{align}
  g_{t}(x)-\mu_{t}(x)
  &=\underbrace{\big(g_{t}(x)-\mathbf{k}_{t}(x)^{\mathrm{T}}
  \mathbf{M}_{t}^{-1}\mathbf{g}_{t}\big)
  -\mathbf{k}_{t}(x)^{\mathrm{T}}\mathbf{M}_{t}^{-1}
  \boldsymbol{\epsilon}_{t}}_{\text{static error}}\nonumber\\
  &\qquad-\mathbf{k}_{t}(x)^{\mathrm{T}}\mathbf{M}_{t}^{-1}\mathbf{d}_{t}.
  \label{eq:err-decomp}
\end{align}
The bracketed terms constitute the estimation error of a static kernel ridge regression with target $g_{t}$ and $|\mathcal{W}_{t}|\le w$ observations. The self-normalized confidence bound~\cite{chowdhury2017}, applied at slot $t$ with level $\delta/(2T)$, bounds them by $\beta\sigma_{t}(x)$ for all $x$ with probability at least $1-\delta/(2T)$, where $\beta$ absorbs the level through the $\ln(2T/\delta)$ term; a union bound over $t\le T$ gives $1-\delta/2$ simultaneously.

Regarding the staleness term, the Cauchy--Schwarz inequality in the metric induced by $\mathbf{M}_{t}$ gives $|\mathbf{k}_{t}(x)^{\mathrm{T}}\mathbf{M}_{t}^{-1}\mathbf{d}_{t}|
\le\|\mathbf{M}_{t}^{-1/2}\mathbf{k}_{t}(x)\|_{2} \|\mathbf{M}_{t}^{-1/2}\mathbf{d}_{t}\|_{2}$. The first factor satisfies $\|\mathbf{M}_{t}^{-1/2}\mathbf{k}_{t}(x)\|_{2}^{2} =k(x,x)-\sigma_{t}^{2}(x)\le\sigma_{f}^{2}$ by \eqref{eq:posterior_var}. For the second, $\mathbf{M}_{t}\succeq\varrho\mathbf{I}$ gives $\|\mathbf{M}_{t}^{-1/2}\mathbf{d}_{t}\|_{2}\le\varrho^{-1/2}\|\mathbf{d}_{t}\|_{2}$, and $|d_{\tau}|\le\|g_{\tau}-g_{t}\|_{\infty}\le V_{\mathcal{W}_{t}}$ for every $\tau\in\mathcal{W}_{t}$ by telescoping, so $\|\mathbf{d}_{t}\|_{2}\le\sqrt{w}\,V_{\mathcal{W}_{t}}$. Applying the triangle inequality to \eqref{eq:err-decomp} establishes \eqref{eq:conf}, and \eqref{eq:staleness-sum} follows from $\sum_{t}V_{\mathcal{W}_{t}}\le wV_{T}$.
\end{IEEEproof}

\subsection{Uncertainty Sums}
\label{app:sums}

The two sums below are the forward and the reverse direction of Lemma~\ref{lem:chain}, restricted to a block of $w$ samples: the first conditions each position on its past and governs how well the learner estimates, the second conditions it on its future and governs how far it travels.

\begin{lemma}[Uncertainty sums]
\label{lem:optsum}
For any sequence $\{x_{t}\}$, adaptively chosen or otherwise, let $\tilde{\sigma}_{t}$ denote the posterior standard deviation given the samples at $\{x_{\tau+1},\dots,x_{t+1}\}$ with $\tau=t-w+1$. Then
\begin{align}
  \sum_{t=1}^{T}\sigma_{t}^{2}(x_{t+1})
    &=O\big(T\gamma_{w}/w\big),
  \label{eq:optsum}\\
  \sum_{t=1}^{T}\sigma_{t}(x_{t+1})
  \;=\;\sum_{t=1}^{T}\tilde{\sigma}_{t}(x_{t-w+1})
    &=O\big(T\sqrt{\gamma_{w}/w}\big).
  \label{eq:backward}
\end{align}
\end{lemma}
\begin{IEEEproof}
Partition into $\lceil T/w\rceil$ consecutive blocks of exactly $w$, the last possibly shorter. For the forward sum, fix a block $\mathcal{B}$ with first slot $t_{\mathcal{B}}$: for $t\in\mathcal{B}$ we have $t_{\mathcal{B}}\ge t-w+1$, so the samples of $\mathcal{B}$ collected up to $t$ lie in $\mathcal{W}_{t}$, and conditioning on additional data cannot increase the posterior variance~\cite{rasmussen2006}, whence $\sigma_{t}(x_{t+1})\le\sigma_{t_{\mathcal{B}}:t}(x_{t+1})$. For the reverse sum, fix a block with last index $t_{\mathcal{B}}$: the conditioning set of $\tilde{\sigma}_{t}$ with $t=\tau+w-1$ is $\{\tau+1,\dots,\tau+w\}$, which contains the samples of $\mathcal{B}$ following $\tau$, whence $\tilde{\sigma}_{t}(x_{\tau})\le \sigma_{\tau+1:t_{\mathcal{B}}}(x_{\tau})$. Both majorants are the quantities appearing in the two directions of \eqref{eq:chain-sum} applied to the $w$ positions of the block, each $O(\gamma_{w})$ when summed over it. Summing over the blocks gives \eqref{eq:optsum}, and the Cauchy--Schwarz inequality $\sum_{t}\sigma_{t}\le(T\sum_{t}\sigma_{t}^{2})^{1/2}$, applied to each, gives \eqref{eq:backward}.
\end{IEEEproof}

\subsection{Surface Drift}
\label{app:surface}

Lemma~\ref{lem:switch} bounds the switch count by $S_{T}$, which the following lemma bounds in turn. Let $E_{\epsilon}$ denote the event $\sum_{t\le T}\epsilon_{t}^{2}=O\big(R^{2}(T+\ln(2/\delta))\big)$, which holds with probability at least $1-\delta/2$ by Bernstein's inequality, since $\epsilon_{t}^{2}$ is conditionally sub-exponential under Assumption~\ref{as:noise}, and write
$a\triangleq\sigma_{f}\big(\varrho^{-1}(\beta\sigma_{f}+R) +\varrho^{-1/2}\beta\big)$.

\begin{lemma}[Surface drift]
\label{lem:surface}
On $E_{\epsilon}$ and the confidence event of Lemma~\ref{lem:confidence},
\begin{equation}
  S_{T}=O\Big(a\,T\sqrt{\gamma_{w}/w}
  +\frac{\sigma_{f}^{3}}{\varrho^{3/2}}\,w^{3/2}V_{T}\Big).
  \label{eq:surface-bound-lem}
\end{equation}
\end{lemma}
\begin{IEEEproof}
Write $F_{t+1}(x)-F_{t}(x)=[\mu_{t+1}(x)-\mu_{t}(x)]+\beta[\sigma_{t+1}(x)-\sigma_{t}(x)]-\tilde{c}[|x-x_{t+1}|-|x-x_{t}|]$. The third difference is at most $\tilde{c}|x_{t+1}-x_{t}|\le\tilde{c}\rho_{\max}$ uniformly in $x$ by the reverse triangle inequality, and by \eqref{eq:price} $\tilde{c}\rho_{\max}T=(c_{\mathrm{e}}\rho_{\max}+\eta\sigma_{f})T/w=O(T/w)$, which is dominated by the first term of \eqref{eq:surface-bound-lem}. The other two are perturbations of a kernel ridge estimator under the exchange of a single sample, $\mathcal{W}_{t+1}=\mathcal{V}\setminus\{\tau\}$ with $\mathcal{V}\triangleq\mathcal{W}_{t}\cup\{t+1\}$ and $\tau\triangleq t-w+1$. Throughout, $\tilde{\sigma}$ denotes the posterior standard deviation given $\mathcal{W}_{t+1}$, the quantity of Lemma~\ref{lem:optsum}, and we use $|k_{A}(x,x')|\le\sigma_{A}(x)\sigma_{A}(x')\le\sigma_{f}\sigma_{A}(x')$.

Adjoining a sample cannot increase the posterior variance, so $\sigma_{\mathcal{V}}\le\sigma_{t}$ pointwise. The rank-one update of the posterior, read from $\mathcal{W}_{t+1}$ to $\mathcal{V}$, gives
\begin{equation}
  \sigma_{t+1}^{2}(x)-\sigma_{\mathcal{V}}^{2}(x)
  =\frac{k_{\mathcal{W}_{t+1}}(x,x_{\tau})^{2}}
  {\tilde{\sigma}^{2}(x_{\tau})+\varrho}
  \le\frac{\sigma_{f}^{2}}{\varrho}\,\tilde{\sigma}^{2}(x_{\tau}),
  \label{eq:var-update}
\end{equation}
and the same device applied to the insertion gives $\sigma_{t}^{2}(x)-\sigma_{\mathcal{V}}^{2}(x)\le\varrho^{-1}\sigma_{f}^{2}\sigma_{t}^{2}(x_{t+1})$. Since $\sigma_{\mathcal{V}}$ is at most both $\sigma_{t}$ and $\sigma_{t+1}$, whichever of the two is larger exceeds the other by at most its own excess over $\sigma_{\mathcal{V}}$, and $\alpha-\alpha'\le\sqrt{\alpha^{2}-\alpha'^{2}}$ for $\alpha\ge\alpha'\ge0$, so
\begin{equation}
  \sup_{x}\big|\sigma_{t+1}(x)-\sigma_{t}(x)\big|
  \le\frac{\sigma_{f}}{\sqrt{\varrho}}
  \max\big\{\tilde{\sigma}(x_{\tau}),\sigma_{t}(x_{t+1})\big\}.
  \label{eq:var-bound}
\end{equation}

Adjoining an observation $(y,r)$ to a set $A$ moves the posterior mean by
\begin{equation}
  \mu_{A\cup\{y\}}(x)-\mu_{A}(x) =\frac{k_{A}(x,y)}{\sigma_{A}^{2}(y)+\varrho}\big(r-\mu_{A}(y)\big),
  \label{eq:mean-update}
\end{equation}
a correlation factor at most $\sigma_{f}\sigma_{A}(y)/\varrho$ times the prediction residual at $y$. For the insertion, $A=\mathcal{W}_{t}$ and $y=x_{t+1}$, and the confidence event bounds the residual by $\beta\sigma_{f}+S_{t}+v_{t}+|\epsilon_{t+1}|$ with $v_{t}\triangleq\|g_{t+1}-g_{t}\|_{\infty}$. For the deletion, read \eqref{eq:mean-update} from $A=\mathcal{W}_{t+1}$ to $\mathcal{V}$ with $y=x_{\tau}$: the factor is at most $\sigma_{f}\tilde{\sigma}(x_{\tau})/\varrho$ and the residual is at most $\beta\sigma_{f}+S_{t+1}+V_{\mathcal{W}_{t+1}}+|\epsilon_{\tau}|$, the drift term now spanning the window.

Both \eqref{eq:var-bound} and the mean bound are linear in an uncertainty at a sampled position, at $x_{t+1}$ given the past and at $x_{\tau}$ given the future, both summed by Lemma~\ref{lem:optsum} to $O(T\sqrt{\gamma_{w}/w})$. The variance terms contribute $\beta\sigma_{f}\varrho^{-1/2}$ times that, and the $\beta\sigma_{f}$ terms of the mean bound $\varrho^{-1}\beta\sigma_{f}^{2}$ times it. For the noise terms, Cauchy--Schwarz and $E_{\epsilon}$ give $\sum_{t}\sigma_{t}(x_{t+1})|\epsilon_{t+1}|
\le(\sum_{t}\sigma_{t}^{2}(x_{t+1}))^{1/2} (\sum_{t}\epsilon_{t+1}^{2})^{1/2}=O(RT\sqrt{\gamma_{w}/w})$, identically for the deletion terms, so these contribute $\varrho^{-1}\sigma_{f}R$ times it; collecting the three gives $a$. The staleness terms contribute $\varrho^{-1}\sigma_{f}^{2}\sum_{t}S_{t}=O(\varrho^{-3/2}\sigma_{f}^{3}w^{3/2}V_{T})$ by \eqref{eq:staleness-sum}, which is the second term of \eqref{eq:surface-bound-lem} and absorbs the drift terms $O(\varrho^{-1}\sigma_{f}^{2}wV_{T})$.
\end{IEEEproof}

\subsection{Dwell Regret}
\label{app:dwell}

The persistence rule never lets the target fall more than $\Delta$ behind the best available position, and it is this invariant, rather than exact optimality of $z_{t}$, that the optimism argument consumes.

\begin{proposition}[Persistence invariant]
\label{prop:persist}
For every $t$, $F_{t}(z_{t})\ge\max_{x\in\mathcal{D}}F_{t}(x)-\Delta$.
\end{proposition}
\begin{IEEEproof}
If the condition in \eqref{eq:hysteresis} fails, $z_{t}$ is the maximizer of $F_{t}$ and the gap is zero; if it holds, $z_{t}=z_{t-1}$ and the condition is the claim.
\end{IEEEproof}

\begin{lemma}[Dwell regret]
\label{lem:dwell}
On the event of Lemma~\ref{lem:confidence}, every dwell slot $t$ and every $\tilde{x}\in\mathcal{D}$ satisfy
\begin{equation}
  g_{t}(\tilde{x})-g_{t}(x_{t+1})
  \le2\beta\sigma_{t}(x_{t+1})+2S_{t}+\tilde{c}D+\Delta.
  \label{eq:dwell}
\end{equation}
\end{lemma}
\begin{IEEEproof}
By the upper bound in \eqref{eq:conf}, $g_{t}(\tilde{x})\le\mathrm{UCB}_{t}(\tilde{x})+S_{t}$. Since $\tilde{x}\in\mathcal{D}$, Proposition~\ref{prop:persist} applied at $\tilde{x}$ gives $\mathrm{UCB}_{t}(z_{t})-\tilde{c}|z_{t}-x_{t}|\ge\mathrm{UCB}_{t}(\tilde{x})-\tilde{c}|\tilde{x}-x_{t}|-\Delta$, and since $\tilde{c}|z_{t}-x_{t}|\ge0$ while $|\tilde{x}-x_{t}|\le D$, this gives $\mathrm{UCB}_{t}(\tilde{x})\le\mathrm{UCB}_{t}(z_{t})+\tilde{c}D+\Delta$. On a dwell slot $x_{t+1}=z_{t}$, and the lower bound in \eqref{eq:conf} applied at $x_{t+1}$ yields $\mathrm{UCB}_{t}(x_{t+1})\le g_{t}(x_{t+1})+2\beta\sigma_{t}(x_{t+1})+S_{t}$. Chaining the three gives \eqref{eq:dwell}.
\end{IEEEproof}


\bibliographystyle{IEEEtran}
\bibliography{lymarl_refs_journal}

@inproceedings{srinivas2010,
  author    = {Srinivas, Niranjan and Krause, Andreas and Kakade,
               Sham M. and Seeger, Matthias},
  title     = {Gaussian Process Optimization in the Bandit Setting:
               No Regret and Experimental Design},
  booktitle = {Proc. 27th Int. Conf. Machine Learning (ICML)},
  address   = {Haifa, Israel},
  pages     = {1015--1022},
  month     = jun,
  year      = {2010}
}

@inproceedings{chowdhury2017,
  author    = {Chowdhury, Sayak Ray and Gopalan, Aditya},
  title     = {On Kernelized Multi-Armed Bandits},
  booktitle = {Proc. 34th Int. Conf. Machine Learning (ICML)},
  address   = {Sydney, Australia},
  pages     = {844--853},
  month     = aug,
  year      = {2017}
}

@article{zhu2024movable,
  author  = {L. Zhu and W. Ma and R. Zhang},
  title   = {Movable Antennas for Wireless Communication:
             Opportunities and Challenges},
  journal = {IEEE Commun. Mag.},
  volume  = {62},
  number  = {6},
  pages   = {114--120},
  year    = {2024}
}

@article{wong2021fluid,
  author  = {K.-K. Wong and A. Shojaeifard and K.-F. Tong and
             Y. Zhang},
  title   = {Fluid Antenna Systems},
  journal = {IEEE Trans. Wireless Commun.},
  volume  = {20},
  number  = {3},
  pages   = {1950--1962},
  year    = {2021}
}

@article{ma2024multibeam,
  author  = {W. Ma and L. Zhu and R. Zhang},
  title   = {Multi-Beam Forming with Movable-Antenna Array},
  journal = {IEEE Commun. Lett.},
  volume  = {28},
  number  = {3},
  pages   = {697--701},
  year    = {2024}
}

@article{feng2024weighted,
  author  = {B. Feng and Y. Wu and X.-G. Xia and C. Xiao},
  title   = {Weighted Sum-Rate Maximization for Movable
             Antenna-Enhanced Wireless Networks},
  journal = {IEEE Wireless Commun. Lett.},
  volume  = {13},
  number  = {6},
  pages   = {1770--1774},
  year    = {2024}
}

@article{ning2025movable,
  author  = {B. Ning and S. Yang and Y. Wu and P. Wang and
             W. Mei and C. Yuen and E. Bj{\"o}rnson},
  title   = {Movable Antenna-Enhanced Wireless Communications:
             General Architectures and Implementation Methods},
  journal = {IEEE Wireless Commun.},
  volume  = {32},
  number  = {5},
  pages   = {108--116},
  year    = {2025}
}

@article{li2026trajectory,
  title={Trajectory optimization for minimizing movement delay in movable antenna systems},
  author={Li, Qingliang and Mei, Weidong and Zhang, Rui and Ning, Boyu},
  journal={IEEE Trans. Wireless Commun.},
  volume={25},
  pages={6986--6999},
  year={2026},
  publisher={IEEE}
}

@inproceedings{bogunovic2016,
  title={Time-varying Gaussian process bandit optimization},
  author={Bogunovic, Ilija and Scarlett, Jonathan and Cevher, Volkan},
  booktitle={Artificial Intelligence and Statistics},
  pages={314--323},
  year={2016},
  organization={PMLR}
}

@article{koren2017nips,
  title={Multi-armed bandits with metric movement costs},
  author={Koren, Tomer and Livni, Roi and Mansour, Yishay},
  journal={Advances in Neural Information Processing Systems},
  volume={30},
  year={2017}
}

@article{zeng2025csi,
  title={{CSI}-free position optimization for movable antenna communication systems: A derivative-free optimization approach},
  author={Zeng, Xianlong and Fang, Jun and Wang, Bin and Ning, Boyu and Li, Hongbin},
  journal={IEEE Wireless Commun. Lett.},
  volume={14},
  number={1},
  pages={53--57},
  year={2025},
  publisher={IEEE}
}

@article{wang2026throughput,
  author  = {H. Wang and Q. Wu and Y. Gao and W. Chen and
             W. Mei and G. Hu and L. Xu},
  title   = {Throughput Maximization for Movable Antenna Systems
             with Movement Delay Consideration},
  journal = {IEEE Trans. Wireless Commun.},
  volume  = {25},
  pages   = {883--899},
  year    = {2026}
}

@article{ding2026energy,
  title={Energy efficiency maximization for movable antenna communication systems},
  author={Ding, Jingze and Zhou, Zijian and Zhu, Lipeng and Zhao, Yuping and Jiao, Bingli and Zhang, Rui},
  journal={IEEE Trans. Wireless Commun.},
  volume={25},
  pages={2624--2638},
  year={2026},
  publisher={IEEE}
}

@article{cheng2024sumrate,
  author  = {Z. Cheng and N. Li and J. Zhu and X. She and
             C. Ouyang and P. Chen},
  title   = {Sum-Rate Maximization for Fluid Antenna Enabled
             Multiuser Communications},
  journal = {IEEE Commun. Lett.},
  volume  = {28},
  number  = {5},
  pages   = {1206--1210},
  year    = {2024}
}

@inproceedings{tae2026delay,
  author    = {Tae, Yunseob and Lee, Jeongjae and Hong, Songnam},
  title     = {Movement Delay-Aware Sum-Rate Maximization for Movable Antenna Systems},
  booktitle = {Proc. IEEE Veh. Technol. Conf. (VTC)},
  address   = {Boston, USA},
  month     = may,
  year      = {2026},
  note      = {to appear}
}

@inproceedings{cheung2019,
  author    = {W. C. Cheung and D. Simchi-Levi and R. Zhu},
  title     = {Learning to Optimize under Non-Stationarity},
  booktitle = {Proc. Int. Conf. Artif. Intell. Statist.
               (AISTATS)},
  pages     = {1079--1087},
  year      = {2019}
}

@article{cesabianchi2013,
  author  = {N. Cesa-Bianchi and O. Dekel and O. Shamir},
  title   = {Online Learning with Switching Costs and Other
             Adaptive Adversaries},
  journal = {Adv. Neural Inf. Process. Syst. (NeurIPS)},
  year    = {2013}
}

@article{besbes2019,
  author  = {Besbes, Omar and Gur, Yonatan and Zeevi, Assaf},
  title   = {Optimal Exploration--Exploitation in a Multi-Armed
             Bandit Problem with Non-Stationary Rewards},
  journal = {Stochastic Systems},
  volume  = {9},
  number  = {4},
  pages   = {319--337},
  year    = {2019}
}

@book{rasmussen2006,
  author    = {Rasmussen, Carl Edward and Williams, Christopher
               K. I.},
  title     = {Gaussian Processes for Machine Learning},
  publisher = {MIT Press},
  address   = {Cambridge, MA},
  year      = {2006}
}

@article{zhu2023modeling,
  title={Modeling and performance analysis for movable antenna enabled wireless communications},
  author={Zhu, Lipeng and Ma, Wenyan and Zhang, Rui},
  journal={IEEE Trans. Wireless Commun.},
  volume={23},
  number={6},
  pages={6234--6250},
  month={Jun.},
  year={2024},
  publisher={IEEE}
}

\end{document}